\documentclass[11pt]{article}
\usepackage[margin=1in]{geometry}
\usepackage[T1]{fontenc}
\usepackage[scale=0.97]{XCharter} 
\usepackage[libertine,bigdelims,vvarbb,scaled=1.05]{newtxmath} 
 
\usepackage{babel}
\usepackage[spacing=true,kerning=true,babel=true,tracking=true]{microtype}
\usepackage[numbers,comma,sort&compress]{natbib}
\usepackage{authblk}
\usepackage[font=small]{caption}
\usepackage[labelformat=simple]{subcaption}
\usepackage{placeins}
\usepackage{amsmath, amsthm}
\usepackage{bm}
\usepackage{braket}
\usepackage{comment}
\usepackage{dsfont}
\usepackage{enumitem}
\usepackage{float}
\usepackage{footnote}
\usepackage{graphicx}
\usepackage{mathtools}
\usepackage{mathrsfs}
\usepackage{multirow}
\usepackage{physics}  
\usepackage{textgreek}
\usepackage[usenames,dvipsnames]{xcolor}
\definecolor{Gred}{RGB}{219, 50, 54}
\definecolor{ToCgreen}{RGB}{0, 128, 0}
\usepackage[colorlinks]{hyperref}
\usepackage{cleveref}
\hypersetup{
	colorlinks=true,
	citecolor=ToCgreen,
	linkcolor=Sepia,
	filecolor=Gred,
	urlcolor=Gred
}

\crefname{theorem}{theorem}{theorems}
\Crefname{theorem}{Theorem}{Theorems}

\crefname{fact}{fact}{facts}
\Crefname{fact}{Fact}{Facts}

\crefname{definition}{definition}{definitions}
\Crefname{definition}{Definition}{Definitions}
\crefname{lemma}{lemma}{lemmas}
\Crefname{lemma}{Lemma}{Lemmas}
\crefname{claim}{claim}{claims}
\Crefname{claim}{Claim}{Claims}
\crefname{proposition}{proposition}{propositions}
\Crefname{proposition}{Proposition}{Propositions}
\crefname{example}{example}{examples}
\Crefname{example}{Example}{Examples}
\crefname{corollary}{corollary}{corollaries}
\Crefname{corollary}{Corollary}{Corollaries}
\crefname{remark}{remark}{remarks}
\Crefname{remark}{Remark}{Remarks}
\crefname{problem}{problem}{problems}
\Crefname{problem}{Problem}{Problems}
\crefname{conjecture}{conjecture}{conjectures}
\Crefname{conjecture}{Conjecture}{Conjectures}
\crefname{task}{task}{tasks}
\Crefname{task}{Task}{Tasks}

\newtheorem{theorem}{Theorem}
\newtheorem{fact}[theorem]{Fact}
\newtheorem{definition}[theorem]{Definition}
\newtheorem{lemma}[theorem]{Lemma}

\newtheorem{corollary}[theorem]{Corollary}

\newtheorem{task}[theorem]{Task}

\newcommand{\R}{\mathbb R}
\newcommand{\N}{\mathbb N}

\renewcommand{\H}{\mathbb H}

\newcommand{\cA}{\mathcal{A}}

\newcommand{\cE}{\mathcal{E}}

\newcommand{\cM}{\mathcal{M}}

\newcommand{\cO}{\mathcal{O}}

\newcommand{\ba}{\boldsymbol{a}}

\newcommand{\bp}{\boldsymbol{p}}
\newcommand{\bq}{\boldsymbol{q}}
\newcommand{\br}{\boldsymbol{r}}

\newcommand{\Id}{\mathds{1}}
\newcommand{\Zero}{\mathds{O}}

\newcommand{\E}{\mathop{\mathbb{E}\/}}

\renewcommand{\tr}{\mathrm{tr}}

\DeclareMathOperator{\poly}{poly}

\newcommand{\tO}{\tilde{\mathcal{O}}}

\newcommand{\Haar}{\text{H}}
\newcommand{\SWAP}{\textnormal{SWAP}}

\DeclarePairedDelimiter\parens{\lparen}{\rparen}
\DeclarePairedDelimiter\floor{\lfloor}{\rfloor}
\DeclarePairedDelimiter\ceil{\lceil}{\rceil}
\DeclarePairedDelimiter\bracks{\lbrack}{\rbrack}

\renewcommand{\epsilon}{\varepsilon}
\renewcommand{\emptyset}{\varnothing}
\DeclareMathAlphabet{\pazocal}{OMS}{zplm}{m}{n}
\renewcommand{\mathcal}[1]{\pazocal{#1}}

\allowdisplaybreaks

\begin{document}

\title{Exponential Advantage from One More Replica in Estimating Nonlinear Properties of Quantum States}

\author[1,2]{Qi Ye\thanks{Email: \href{mailto:yeq22@mails.tsinghua.edu.cn}{yeq22@mails.tsinghua.edu.cn}}}
\author[3]{Zhenhuan Liu}
\author[1,2,4]{Dong-Ling Deng\thanks{Email: \href{mailto:dldeng@tsinghua.edu.cn}{dldeng@tsinghua.edu.cn}}}

\affil[1]{Center for Quantum Information, IIIS, Tsinghua University, Beijing, China}
\affil[2]{Shanghai Qi Zhi Institute, Shanghai, China}
\affil[3]{Yau Mathematical Sciences Center, Tsinghua University, Beijing 100084, China}
\affil[4]{Hefei National Laboratory, Hefei, China}

\date{}
\maketitle

\begin{abstract}
Inferring nonlinear features of quantum states is fundamentally important across quantum information science, but remains challenging due to the intrinsic linearity of quantum mechanics. 
It is widely recognized that quantum memory and coherent operations help avoid exponential sample complexity, by mapping nonlinear properties onto linear observables over multiple copies of the target state. 
In this work, we prove that such a conversion is not only sufficient but also necessary. 
Specifically, we prove that the estimation of $\mathrm{tr}(\rho^{k} O)$ for a broad class of observables $O$ is exponentially hard for any protocol restricted to $(k-1)$-replica joint measurements, whereas access to one additional replica reduces the complexity to a constant. 
These results establish, for the first time, an exponential separation between $(k-1)$- and $k$-replica protocols for any integer $k>2$, thereby introducing a fine-grained hierarchy of replica-based quantum advantage and resolving an open question in the literature. 
The technical core is a general indistinguishability principle showing that any ensemble constructed from large Haar random states via tensor products and mixtures is hard to distinguish from its average.
Leveraging this principle, we further prove that $k$-replica joint measurements are also necessary for distinguishing rank-$k$ density matrices from rank-$(k-1)$ ones. 
Overall, our work delineates sharp boundaries on the power of joint measurements, highlighting resource--complexity trade-offs in quantum learning theory and deepening the understanding of quantum mechanics' intrinsic linearity.
\end{abstract}

\thispagestyle{empty}
\clearpage
\thispagestyle{empty}
\tableofcontents
\thispagestyle{empty}
\clearpage
\newpage
\addtocounter{page}{-2}

\section{Introduction}
Quantum learning is a fundamental process in quantum information science that bridges the gap between quantum and classical worlds. 
Exploiting resources and strategies for quantum learning is thus crucial for reducing the complexity of quantum benchmarking and developing efficient quantum algorithms~\cite{badescuQuantumStateCertification2019,Eisert2020certification}.
Among all quantum learning tasks, estimating nonlinear functions of quantum states is key to a wide range of fields in quantum information science, including quantum resource certification~\cite{renyi1961measures,guhneEntanglementDetection2009,Leone2022sre}, quantum algorithms~\cite{lloydQuantumPrincipalComponent2014}, quantum metrology~\cite{giovannettiAdvancesQuantumMetrology2011}, and quantum simulation~\cite{cotlerQuantumVirtualCooling2019,koczorExponentialErrorSuppression2021,hugginsVirtualDistillationQuantum2021}.
Despite their broad utility, accessing nonlinear functions is challenging as the linearity of quantum measurements prevents direct measurement of nonlinear functions.

A standard strategy to circumvent this obstacle is to exploit multiple replicas of the state and joint quantum operations among them. 
As an example, consider a typical nonlinear quantity $\tr(\rho^{k}O)$, where $\rho$ is the target state and $O$ is an observable.
This quantity covers state moments by setting $O$ to the identity operator, and it serves as a central quantity in entanglement detection~\cite{guhneEntanglementDetection2009}, quantum virtual cooling~\cite{cotlerQuantumVirtualCooling2019}, quantum error mitigation~\cite{koczorExponentialErrorSuppression2021,hugginsVirtualDistillationQuantum2021}, and detecting quantum phases of matter for mixed states \cite{leeQuantumCriticalityDecoherence2023, zhangProbingMixedstatePhases2025}.
This quantity is nonlinear in $\rho$ for any integer $k>1$, but always linear in $\rho^{\otimes k}$.
Therefore, with the ability to jointly operate on $k$ copies of $\rho$, $\tr(\rho^{k}O)$ can be efficiently estimated by measuring a linear observable.
A well-known example is the generalized swap test \cite{ekertDirectEstimationsLinear2002} which estimates $\tr(\rho^k)$ efficiently using $k$-replica joint measurements. 
However, implementing such $k$-replica joint measurements remains experimentally challenging, requiring either large-scale quantum memory or efficient state resetting.
Demonstrations have so far only reached small values of $k$ \cite{islamMeasuringEntanglementEntropy2015,kaufmanQuantumThermalizationEntanglement2016,cotlerQuantumVirtualCooling2019,huangQuantumAdvantageLearning2022}.
This raises a fundamental and practical question:
\begin{center}
\textit{Are $k$-replica joint measurements necessary for efficient estimation of $\tr(\rho^{k}O)$?}
\end{center}
If the answer is negative, it means that we can design a more efficient learning protocol that utilizes fewer than $k$ copies of $\rho$ at a time to accurately estimate $\tr(\rho^{k}O)$.
If the answer is positive, it uncovers that some nonlinear functions are fundamentally more difficult to estimate than other nonlinear ones.
This would deepen our understanding of the intrinsic linearity of quantum mechanics from the viewpoint of quantum learning.

This question also lies within a broader pursuit of understanding replica quantum advantage, that is, the extent to which multi-replica measurements enhance our ability to learn properties of quantum states. 
Extensive studies have shown that allowing 2-replica joint measurements provides exponential advantages over single-replica protocols in a variety of fundamental tasks, including purity estimation, shadow tomography, and entanglement detection~\cite{aharonovQuantumAlgorithmicMeasurement2022,chenExponentialSeparationsLearning2022,chenOptimalTradeoffsEstimating2024a,gongSampleComplexityPurity2024,kingTriplyEfficientShadow,liu2025separation}. 
Moreover, \cite{chenHierarchyReplicaQuantum2021} established a separation between $\poly(n,k)$-replica and $k$-replica protocols, where $n$ is the qubit number of the target state. 
Yet beyond these specific cases, our understanding remains limited. 
In particular, it is not even known whether 3-replica protocols can be exponentially more powerful than 2-replica protocols for any learning task. An open question raised by \cite{chenHierarchyReplicaQuantum2021} and re-emphasized recently in \cite{focs2024openproblems} is
\begin{center}
\textit{Are there learning tasks which are $(k-1)$-replica hard but $k$-replica easy for $k>2$?}
\end{center}

We answer both questions affirmatively in this work. 
We prove that if one is restricted to joint measurements on fewer than $k$ replicas, then accurately estimating $\tr(\rho^{k} O)$ for a wide variety of observables $O$ requires exponentially many samples. 
In contrast, with access to $k$-replica joint measurements, the task becomes simple and can be achieved with constant sample complexity. 
This establishes, for the first time, a fine-grained hierarchy of replica quantum advantage for any integer $k\ge 2$ \footnote{During the preparation of this manuscript, we became aware of the independent and concurrent work by Nöller, Tran, Gachechiladze, and Kueng~\cite{Noeller2025hierarchy}, which builds separations between $k$- and $(k-1)$- replica protocols for all primes $k$ based on other learning tasks. See \Cref{sec: related works} for details.}, and delineates sharp boundaries on the power of joint measurements.

To establish this separation, we develop a general indistinguishability principle, showing that any ensemble constructed from large Haar-random states via tensor products and mixtures is hard to distinguish from its average. 
We expect this principle to have broad applicability as a tool for proving lower bounds in quantum learning theory. 
As another illustrative applications, we also employ it in spectrum testing and rank testing. 
We prove that $k$-replica joint measurements are likewise necessary to efficiently distinguish two spectra that coincide on all moments up to degree $k-1$, which implies an exponential sample complexity lower bound for testing the rank of an unknown state using bounded-replica joint measurements.

\subsection{Our results}
Here we provide a more detailed introduction to our results. For clarity of exposition, in what follows we focus on $\tr(\rho^{k+1}O)$ rather than $\tr(\rho^k O)$. A $k$-replica protocol is a learning protocol that can only perform joint measurements on at most $k$ copies of the unknown states at a time, see \Cref{sec: k replica} for formal definitions.

\paragraph{A general indistinguishability principle.} A common approach to proving lower bounds in quantum learning theory is to reduce the learning task to a distinguishing problem between two ensembles of quantum states. In a nutshell, if two ensembles are hard to distinguish, then any property that differs significantly between them must also be hard to learn, otherwise one could distinguish the ensembles by learning that property.
A fruitful example is pseudorandom states~\cite{jiPseudorandomQuantumStates}, which are easy to generate yet indistinguishable from Haar-random states by any polynomial-time algorithm under standard computational assumptions. Building on this, prior works have shown that it is computationally hard to test whether a state has imaginary amplitudes~\cite{hinscheEfficientDistributedInner2024}, negative amplitudes~\cite{giurgica-tironPseudorandomnessSubsetStates2023,jeronimoPseudorandomPseudoentangledStates2024}, high entanglement~\cite{boulandQuantumPseudoentanglement2022}, or high magic~\cite{guPseudomagicQuantumStates2024}.
Another example, more closely related to our setting, is to test purity using restricted measurements.
In \cite{chenExponentialSeparationsLearning2022,chenOptimalTradeoffsEstimating2024a}, the authors show the hardness of distinguishing the Haar-random ensemble $\cE=\{\psi\leftarrow \mu_\Haar(d)\}$ from the maximally mixed state $\rho_0=\Id/d$ with 1-replica protocols, thereby implying the hardness of purity estimation.
These examples illustrate a general paradigm that new indistinguishability results often translate directly into new lower bounds for quantum learning tasks.

In this work, we establish an indistinguishability principle that substantially generalizes the purity testing results in~\cite{chenExponentialSeparationsLearning2022,chenOptimalTradeoffsEstimating2024a}. We prove that any ensemble constructed from large Haar-random states via tensor products and mixtures is indistinguishable from its average without joint measurements.

\begin{definition}[Haar-assembled ensembles]\label{def: Haar assembled ensembles}
    For $m,d,D\in \N^+$, a $D$-dimensional Haar-assembled ensemble $\cE_\Haar$ assembled by $m$ $d$-dimensional Haar-random states $\psi_1,\cdots, \psi_m$ has the form
    \begin{equation}
        \cE_\Haar \coloneqq \left\{\sum_{j=1}^J p_j U_j(\psi_1^{\otimes a_{j1}}\otimes \cdots \otimes \psi_m^{\otimes a_{jm}}\otimes \tau_j)U_j^\dagger\right\}_{\psi_1,\cdots, \psi_m \leftarrow \mu_\Haar(d)}\,,\label{equ: Haar assembled ensembles}
    \end{equation}
    where $J\in \N^+$, $\{p_j\}_{j=1}^J$ is a probability distribution. For each $j\in [J]$ and $r\in [m]$, $U_j$ is a fixed unitary, $a_{jr}\in \N$, $\tau_j$ is a fixed state such that the dimension of each term is $D$.
\end{definition}

\begin{theorem}[Indistinguishability principle for Haar-assembled ensembles. Informal version of \Cref{thm: Haar-assembled ensemble vs mean}]\label{thm:indistinguishability_informal}
Any 1-replica protocol requires $\Omega(\sqrt{d/(m^2\max\{a_{jr}\}})$ samples to distinguish a Haar-assembled ensemble $\cE_\Haar$ from its average $\E_{\rho\leftarrow \cE_\Haar}[\rho]$ with a constant success probability.
\end{theorem}

\noindent Compared to \cite{chenExponentialSeparationsLearning2022,chenOptimalTradeoffsEstimating2024a}, we now allow an arbitrary number of independent Haar-random states, arbitrary multiplicities in the tensor product, and arbitrary mixtures of such tensor products. Given its flexibility, we expect the indistinguishability principle to have broad applications in proving lower bounds for quantum learning. In this paper, we apply it in particular to lower bounds for nonlinear property estimation, spectrum testing, and rank testing.

\paragraph{Lower bounds for estimating $\tr(\rho^{k+1}O)$.} With Theorem~\ref{thm:indistinguishability_informal}, we can derive our main result that establishes the hardness of estimating $\tr(\rho^{k+1}O)$ using $k$-replica protocols.
\begin{theorem}\label{thm: main theorem}
    Let $k, d\in \N^+$, $O$ be a $d$-dimensional observable with the operator norm $\norm{O}_{\infty}\leq 1$ and $\norm{O}_1\ge 2(k+1)^{2k}$. Let $\epsilon$ be the error parameter such that $50(k+1)^3d^{-1/2}\leq \epsilon \leq (2(k+1))^{-k}\norm{O}_1/(10d)$. 
    Given access to a $d$-dimensional mixed quantum state $\rho\in \mathcal{C}^{d\times d}$, the sample complexity of estimating $\tr(\rho^{k+1}O)$ within additive error $\epsilon$ with probability $0.9$ is at least $\Omega(\frac{\sqrt{d}}{k\epsilon^{1/(1+k)}})$ for any $k$-replica protocol.
\end{theorem}

To interpret the theorem, let us focus on the case $k=O(1)$ and $\norm{O}_1=\Omega(d)$, which covers the important case where $O$ is a Pauli string, and in particular the identity observable $O=\Id$. The theorem states that estimating $\tr(\rho^{k+1}O)$ within constant additive error requires $\Omega(\sqrt{d})$ samples using $k$-replica protocols. 
More generally, our hardness result applies to estimating $\tr(\rho^{\otimes (k+1)} O')$ for any $O'$ that has a nontrivial $(k+1)$-body component (see \Cref{sec: Haar measure} for definitions and \Cref{sec: k+1 body observables} for the full statement).
In sharp contrast, with access to $(k+1)$ replicas, the estimation becomes efficient because $\tr(\rho^{k+1}O)$ is a linear function of $\rho^{\otimes k+1}$. The sample complexity is $\cO(k/\epsilon^2)$ (see \Cref{sec: subroutines}). 
Therefore, we obtain an exponential separation between $k$-replica protocols and $(k+1)$-replica protocols for any constant $k$.

\begin{corollary}[Hierarchy of quantum replica advantage]\label{cor: hierarchy}
For $\epsilon=(2(k+1))^{-k}/10$, $\norm{O}_\infty=1$, and $\norm{O}_1=d=2^n$, estimating $\tr(\rho^{k+1}O)$ within error $\epsilon$ requires $\Omega(\sqrt{d})$ samples for $k$-replica protocols but only $\mathcal{O}(k(2(k+1))^{2k})$ samples for $(k+1)$-replica protocols. This yields a constant-vs-exponential separation for constant $k$, polynomial-vs-exponential for $k=\mathcal{O}(\log n/\log\log n)$, and subexponential-vs-exponential for $k=o(n/\log n)$.
\end{corollary}
\Cref{cor: hierarchy} establishes the first fine-grained hierarchy for quantum replica advantage for $k$ up to $o(n/\log n)$, resolving an open problem raised in \cite{chenHierarchyReplicaQuantum2021, focs2024openproblems}.

\paragraph{Upper bounds for estimating $\tr(\rho^{k+1}O)$.}
To study the optimality of our lower bounds, we also prove an upper bound result in \Cref{sec: upper bounds} for estimating $\tr(\rho^{k+1}O)$, which shows that $\lceil (k+1)/2 \rceil$-replica measurements already suffice to match the $\sqrt{d}$ dependence in the lower bound and thus the $\sqrt{d}$ dependence in \Cref{thm: main theorem} is optimal.
\begin{theorem}\label{thm: upper bound} 
Let $k,d\in \N^+$, $\epsilon>0$, and $O$ be a $d$-dimensional observable with $\norm{O}_\infty\leq 1$. There exists a $\ceil{\frac{k+1}{2}}$-replica protocol that estimates $\tr(\rho^{k+1}O)$ within error $\epsilon$ with sample complexity $\tilde{\mathcal{O}}\left(\max\{\frac{k\sqrt{d}}{\epsilon}, \frac{k}{\epsilon^2}\}\right)$.
\end{theorem}

Another possible concern about \Cref{thm: main theorem} is that when $k$ is large, the lower bound only works in the high-precision case ($\epsilon\leq (2(k+1))^{-k}$). However, the upper bound result in \cite{shinResourceefficientAlgorithmEstimating2025a} shows that a high-precision restriction is necessary. Specifically, they prove that when $\epsilon\ge 2(k+1)e^{-k}$, we can infer $\tr(\rho^{k+1})$ from estimators of $\{\tr(\rho^{i})\}_{i=1}^k$, which can be obtained efficiently using $k$-replica protocols.

\paragraph{Lower bounds for spectrum testing}
As another illustrative application of the indistinguishability principle, we obtain a similar lower bound for spectrum testing~\cite{montanaroSurveyQuantumProperty2016,donnell2015spectrum, wrightHowLearnQuantum}.
Let $\bp=(p_1,\dots, p_d)$ and $\bq=(q_1,\dots, q_d)$ be two probability distributions. The $\bp$-versus-$\bq$ spectrum testing~\cite{pelecanosBeatingFullState2025} task asks one to determine whether the spectrum of an unknown state $\rho$ is $\bp$ or $\bq$, promised to be one of the two. 
Depending on the choice of $\bp$ and $\bq$, this task is related to several previously studied problems.
For example, when $\bq=(1/d, \cdots, 1/d)$ is the uniform distribution, this problem is closely related to the mixedness testing \cite{donnell2015spectrum, bubeckEntanglementNecessaryOptimal2020, chenTightBoundsQuantum2022} that asks one to certify a maximally mixed state. When $\bp$ and $\bq$ have different numbers of nonzero elements, the task reduces to rank testing \cite{donnell2015spectrum}.

If $\bp$ and $\bq$ agree on all moments up to degree $k$, a natural distinguishing strategy is to estimate the $(k+1)$-th moment, which indeed requires $(k+1)$-replica measurements. The following theorem (proved in \Cref{sec: spectrum testing}) shows that this approach is in fact optimal in terms of the number of replicas:

\begin{theorem}
    Suppose that $\sum_{r=1}^d p_r^i=\sum_{r=1}^d q_r^i$ for all $i\in [k]$, and that both $\bp$ and $\bq$ have at most $m$ nonzero elements. Then the sample complexity of $\bp$-versus-$\bq$ spectrum testing is at least $\Omega(\sqrt{d}/(m\sqrt{\ln m}))$ for any $k$-replica protocol.
\end{theorem}

By carefully designing $\bp$ and $\bq$, this theorem implies various lower bounds for learning spectral properties of unknown states under restricted measurements. For instance, taking $\bp$ and $\bq$ with $k$ and $k+1$ nonzero entries respectively, we obtain the following corollary for rank testing.
\begin{corollary}\label{cor: rank testing}
    Let $k, d\in \N^+$ and $\epsilon=2/(k+1)^3$. Using $k$-replica protocols, $\Omega(\sqrt{d}/(k\sqrt{\ln k}))$ copies are necessary to test whether a $d$-dimensional unknown state $\rho$ has rank at most $k$ or is $\epsilon$-far (in trace distance) from any state with rank at most $k$.
\end{corollary}
Another application is a lower bound for estimating $\tr(\rho^\alpha)$ for a non-integer $\alpha$, a quantity closely related to the Tsallis entropy~\cite{tsallisPossibleGeneralizationBoltzmannGibbs1988} defined as $S_\alpha(\rho)\coloneqq(1-\tr(\rho^\alpha))/(\alpha-1)$.

\begin{corollary}\label{cor: noninteger power}
    Let $d, m, k\in \N^+$ and $\alpha>0$ be a non-integer. Define $\epsilon(\alpha, k, m)$ as the maximum of $\frac12 (\sum_{r=1}^{m}p_r^\alpha -  \sum_{r=1}^m q_r^\alpha)$ over all pairs of distributions $\{p_r\}_{r=1}^m, \{q_r\}_{r=1}^m$ such that $\sum_{r=1}^m p_r^i=\sum_{r=1}^m q_r^i, \forall i\in [k]$. Given copy access of an unknown $d$-dimensional state $\rho$, the sample complexity of estimating $\tr(\rho^{\alpha})$ within error $\epsilon(\alpha, k, m)$ is at least $\Omega(\sqrt{d}/(m\sqrt{\ln m}))$.
\end{corollary}

Previous works on testing rank~\cite{donnell2015spectrum} or estimating non-integer powers of the trace~\cite{liuEstimatingTraceQuantum2025, chenImprovedSampleUpper2025} typically allow arbitrary joint measurements. Consequently, their sample complexities are dimension-independent --- depending only polynomially on the rank and the precision, but not on the dimension $d$. In contrast, our results show that under bounded-replica protocols, the sample complexity necessarily scales as $\Omega(\sqrt{d})$, thereby suffering from an unavoidable exponential dependence on the system size.

\subsection{Overview of techniques}
Here, we provide a high-level discussion of techniques.

\paragraph{Lower bounds for estimating nonlinear properties.} 
A standard approach to proving lower bounds is to reduce estimation to a distinguishing problem between two carefully designed ensembles $\cE_1$ and $\cE_2$.
The construction must satisfy two conditions that
\begin{enumerate}
\item[(a)] no $k$-replica protocol can efficiently distinguish them,
\item[(b)] we can efficiently distinguish them given the ability to estimate $\tr(\rho^{k+1}O)$.
\end{enumerate}
These two conditions together imply the hardness of estimating $\tr(\rho^{k+1}O)$ under $k$-replica protocols.
The first condition implies that $\rho^{\otimes i}$ has the same distribution under both ensembles for every $i\le k$, otherwise one could distinguish them by estimating some linear observable on $\rho^{\otimes i}$. The second condition demands that $\tr(\rho^{k+1}O)$ differ significantly between the two ensembles.

To construct such ensembles, we start with two distributions $\bp=(p_1,\dots,p_m)$ and $\bq=(q_1,\dots,q_m)$ that match on all moments up to degree $k$ but differ as much as possible on the $(k+1)$-th moment, that is, $\sum_{r=1}^mp_r^i=\sum_{r=1}^m q_r^i$ for $i\leq k$ but $\abs{\sum_{r=1}^mp_r^{k+1}-\sum_{r=1}^m q_r^{k+1}}$ is as large as possible.
Using Chebyshev polynomials, we achieve a gap of $2(2(k+1))^{-k}$ in the $(k+1)$-th moment (\Cref{lem: equal moment sequences}). We believe this to be optimal, and in fact it is the only obstacle to extending the lower bound beyond $\epsilon=\mathcal{O}((2(k+1))^{-k})$.
Given $\bp$ and $\bq$, we define
\begin{equation}
    \cE_1\coloneqq \left\{\sum_{r=1}^m p_r\psi_r\right\}_{\psi_1,\cdots, \psi_m\leftarrow \mu_\Haar(d)},\quad \cE_2\coloneqq \left\{\sum_{r=1}^m q_r\psi_r\right\}_{\psi_1,\cdots, \psi_m\leftarrow \mu_\Haar(d)},\label{equ: overview 1}
\end{equation}
where $\psi_1,\dots,\psi_m$ are independent Haar-random pure states. 
To prove condition (a), observe that distinguishing $\cE_1$ and $\cE_2$ with a $k$-replica protocol is equivalent to distinguishing $\cE_1^{(k)}\coloneqq \{\rho^{\otimes k}\}_{\rho\leftarrow \cE_1}$ and $\cE_2^{(k)}\coloneqq \{\rho^{\otimes k}\}_{\rho\leftarrow \cE_2}$ using $1$-replica protocols on $k$ qudits.
A key observation is that $\cE_1^{(k)}$ and $\cE_2^{(k)}$ are both Haar-assembled ensembles satisfying Definition~\ref{def: Haar assembled ensembles}.
Furthermore, $\cE_1^{(k)}$ and $\cE_2^{(k)}$ have the same average.
Indeed, $\E_{\rho\sim \cE_1}[\rho^{\otimes k}]$ is symmetric in the entries of $\bp$, and hence can be expressed solely in terms of the moments of $\bp$ up to degree $k$, which equal those of $\bq$ by construction.
By our indistinguishability principle for Haar-assembled ensembles, both $\cE_1^{(k)}$ and $\cE_2^{(k)}$ are hard to distinguish from this common average using 1-replica protocols, which is defined on $k$ qudits. 
The triangle inequality then implies that $\cE_1^{(k)}$ and $\cE_2^{(k)}$ are indistinguishable under 1-replica protocols, establishing condition (a).

Now we explain the condition (b). Direct calculation gives $\E_{\rho\leftarrow \cE_1}[\tr(\rho^{k+1}O)]-\E_{\rho\leftarrow \cE_2}[\tr(\rho^{k+1}O)] \approx 2(2(k+1))^{-k}\abs{\tr(O)}/d$. Therefore, we can distinguish the two ensembles given estimators of $\tr(\rho^{k+1}O)$ within error $(2(k+1))^{-k}\abs{\tr(O)}/d$. However, this argument fails when $O$ is traceless. The reason is that the positive and negative parts of $O$ cancel out in the calculation. To circumvent this issue, for a possibly traceless $O$, we select $d/2$ eigenvalues of $O$ such that the absolute value of their sum is at least $\norm{O}_1/4$. Then we define $\cE_1,\cE_2$ inside the $(d/2)$-dimensional subspace spanned by these eigenvalues. In this subspace, the $d/2$ eigenvalues will not cancel out completely so the argument above follows.

\paragraph{Indistinguishability principle for Haar-assembled ensembles}
We next sketch the proof of the indistinguishability principle. The simplest Haar-assembled ensemble is $\cE=\{\psi\}_{\psi\sim \mu_\Haar(d)}$, i.e., a single Haar-random state. In \cite{chenExponentialSeparationsLearning2022}, it was shown that any 1-replica protocol requires $\sqrt{d}$ samples to distinguish $\cE$ from its mean, namely the maximally mixed state $\rho_0$. Although this is already a special case, the proof is highly nontrivial and contains the main ideas of the general argument. So we sketch its proof here.

Suppose a protocol applies 1-replica measurements to $T$ copies of $\rho$. These measurements can be combined into a single measurement $\{F_s\}_s$ on $\rho^{\otimes T}$, where each $F_s$ has product form. For every $F_s$, the one-sided likelihood ratio satisfies
\begin{equation}
    \frac{\tr(F_s \E_\psi[\psi^{\otimes T}])}{\tr(F_s \rho_0^{\otimes T})}=\frac{d^T}{d(d+1)\cdots (d+T-1)}\frac{\tr(F_s \sum_{\pi\in S_T}\pi)}{\tr(F_s)}\overset{*}{\ge} \frac{d^T}{d(d+1)\cdots (d+T-1)}\ge 1-\frac{T^2}{d}.
\end{equation}
Here $(*)$ is a property of permutation operators proved in \cite[Lemma 5.12]{chenExponentialSeparationsLearning2022}. Thus, when $T=o(\sqrt{d})$, the outcome distributions under $\E_\psi[\psi^{\otimes T}]$ and under $\rho_0^{\otimes T}$ are statistically close so that no algorithm can distinguish the two cases. This yields the $\Omega(\sqrt{d})$ lower bound.

The general indistinguishability principle follows a similar strategy but with two additional complications. First, in our setting the state is a mixture rather than a pure state, so we must expand $\psi^{\otimes T}$ and bound the likelihood ratio term by term. Second, each term is now a tensor product of multiple Haar-random states, making the Haar integral more involved. To overcome this, we prove a generalized version of $(*)$ in \Cref{lem: permutation inequality}.

\paragraph{Lower bounds for spectrum testing.} We now show how to apply the technique to spectrum testing.
For two spectra $\bp$ and $\bq$ that need to test, we construct $\cE_1$ and $\cE_2$ as in \eqref{equ: overview 1}. They are hard to distinguish for the same reason.
The only problem is that the spectrum of $\sum_r p_r\psi_r$ is not exactly $\bp$ because Haar random states may not be orthogonal, although with high probability they are approximately orthogonal. We address this with a rounding lemma (\Cref{lem: round spectrum}, essentially the Gram-Schmidt procedure) to find a nearby state whose spectrum is exactly $\bp$. This ensures that replacing the ensembles with their rounded versions introduces only a small error.

\paragraph{Upper bounds for estimating nonlinear properties.}
For \Cref{thm: upper bound}, denote $k_1=\ceil{\frac{k+1}{2}}$ and $k_2=\floor{\frac{k+1}{2}}$ so that $k_1+k_2=k+1$. 
Performing a generalized swap test on $\rho^{\otimes k_1}$ and tracing out $(k_1-1)$ registers, the post-measurement state is $\rho_1=(\rho+\rho^{k_1})/(1+\tr(\rho^{k_1}))$ with probability $(1+\tr(\rho^{k_1}))/2 >1/2$. 
The state $\rho_2$ is defined in the same way with $k_2$.
With high probability, we can prepare a sufficient number of copies of $\rho_1$ and $\rho_2$. As a byproduct, we also get estimators for $\tr(\rho^{k_1})$ and $\tr(\rho^{k_2})$ by counting the number of successful preparations. Now we apply protocols in \cite{anshuDistributedQuantumInner2022,duOptimalRandomizedMeasurements2025} to estimate $\tr(\rho_1\rho_2O)$ in a distributed way.
Using the identity $\tr(\rho^{k+1}O)=(1+\tr(\rho^{k_1}))(1+\tr(\rho^{k_2}))\tr(\rho_1\rho_2O)-\tr(\rho^2O)-\tr(\rho^{k_1+1}O)-\tr(\rho^{k_2+1}O)$, we reduce the task of estimating $\tr(\rho^{k+1}O)$ to estimating $\tr(\rho^2O)$, $\tr(\rho^{k_1+1}O)$, and $\tr(\rho^{k_2+1}O)$, all of which involve strictly smaller degrees (except in the base cases $k=1,2,3$, which we handle separately). By applying this reduction one more time, i.e., $k_1+1=\ceil{\frac{k_1+1}{2}}+\floor{\frac{k_1+1}{2}}$ and $k_2+1=\ceil{\frac{k_2+1}{2}}+\floor{\frac{k_2+1}{2}}$, we ensure that all degrees are eventually not larger than $k_1$, which can be estimated directly using $k_1$-replica protocols.
The total sample complexity is $\tilde{\mathcal{O}}\left(\max\{\frac{k\sqrt{d}}{\epsilon}, \frac{k}{\epsilon^2}\}\right)$ due to the distributed estimation of $\tr(\rho_1\rho_2O)$.


\subsection{Future directions}
\paragraph{Error dependence.}
Our lower bound for estimating $\tr(\rho^{k+1}O)$ exhibits a poor error dependence of $\epsilon^{-1/(k+1)}$. Can this be improved? We note that the optimal dependence is unclear even for purity estimation (i.e., $\tr(\rho^2)$). In fact, it is only known that with 1-replica protocols the sample complexity is $\mathcal{O}(\max\{\tfrac{\sqrt{d}}{\epsilon},\tfrac{1}{\epsilon^2}\})$ and $\Omega(\max\{\tfrac{\sqrt{d}}{\sqrt{\epsilon}},\tfrac{1}{\epsilon^2}\})$ \cite{gongSampleComplexityPurity2024}.

\paragraph{Error regime.}
\Cref{thm: main theorem} shows that estimating $\tr(\rho^{k+1})$ is hard when $\epsilon=\mathcal{O}((2(k+1))^{-k})$, while \cite{shinResourceefficientAlgorithmEstimating2025a} shows that it becomes easy when $\epsilon>2(k+1)e^{-k}$. 
A nontrivial gap therefore remains. Pinpointing the precise threshold of this phase transition is an interesting open problem.

\paragraph{More applications in spectrum testing.}
Our application to spectrum testing (\Cref{thm: spectrum testing}) is meaningful only in the regime $m<\sqrt{d}$. The main obstacle is that when $m$ is large, the spectrum of $\sum_{r}p_r\psi_r$ can deviate significantly from $(p_1,\dots,p_d)$. A possible workaround is to consider ensembles of the form $\{U\operatorname{diag}(p_1,\dots,p_d)U^\dagger\}$ for Haar-random unitaries $U$. However, analyzing such ensembles requires heavy use of representation theory and Weingarten calculus. Can the indistinguishability principle be established for this more natural class of ensembles?

\paragraph{Considering state resetting.}
While our results show that $(k+1)$ replicas are necessary for efficiently estimating $\tr(\rho^{k+1}O)$, this does not imply that $n(k+1)$ working qubits are strictly required, since replicas can be loaded sequentially. 
For instance, there exists a 3-replica protocol that estimates $\tr(\rho^3O)$ using only $(2n+1)$ working qubits: one applies a swap test on two replicas, resets one register to a fresh copy, and then applies another swap test. 
Therefore, a practically meaningful question is that, can we identify a quantum learning task, such as the estimation of $\tr(\rho^{\otimes (k+1)}O_k)$ with a general form of $O_k$, that is exponentially hard for all protocols utilizing $nk$ working qubits?

\section{Related works}\label{sec: related works}
\paragraph{Estimating nonlinear properties of quantum states}
The celebrated swap test \cite{buhrmanQuantumFingerprinting2001} estimates $\tr(\rho^2O)$ using a simple quantum circuit consisting of a controlled-SWAP gate and two Hadamard gates. 
Since then, various methods have been developed to estimate $\tr(\rho^kO)$ for general $k$ and to reduce the overhead, including the generalized swap test~\cite{ekertDirectEstimationsLinear2002}, the two-copy test~\cite{subasiEntanglementSpectroscopyDepthtwo2019}, and multivariate trace estimation~\cite{quekMultivariateTraceEstimation2024}.
For a detailed review of existing methods, we refer the reader to \cite{shinResourceefficientAlgorithmEstimating2025a}, which also provides an improved algorithm for the low-rank case.
On the lower-bound side, recent works~\cite{chenSimultaneousEstimationNonlinear2025,zhangMeasuringLessLearn2025} establish a sample complexity of $\Omega(k/\epsilon^2)$, matching the known upper bounds.

Besides the sample-access model, some papers also study the ``purified query access model'', where one additionally has access to a purification of the unknown state or the oracle to prepare the purification~\cite{gilyenDistributionalPropertyTesting2020}. 
This model can sometimes reduce the overhead in predicting nonlinear quantum properties. 
For instance, \cite{chenSimultaneousEstimationNonlinear2025,zhangMeasuringLessLearn2025} show that with purification, one can estimate $\tr(\rho^kO)$ using $\mathcal{O}(\sqrt{k})$ queries, achieving a quadratic improvement. Similarly, \cite{liuExponentialSeparationsQuantum2025} proves that if $\rho$ can be purified with a constant number of ancilla qubits, then $\tr(\rho^kO)$ can be estimated with only a constant number of queries.

More general nonlinear properties have also been investigated. A particularly important class includes the von Neumann and Rényi entropies \cite{acharyaEstimatingQuantumEntropy2020,wangQuantumAlgorithmsEstimating2023, wangTimeEfficientQuantumEntropy2024, gilyenDistributionalPropertyTesting2020, gurSublinearQuantumAlgorithms2021, wangNewQuantumAlgorithms2024}. Recently, \cite{liuEstimatingTraceQuantum2025} studied the sample, query, and time complexity of estimating $\tr(\rho^q)$ for non-integer $q$. Further improvements were obtained in \cite{chenImprovedSampleUpper2025}, which showed that the sample complexity is $\tilde{\Theta}(1/\epsilon^2)$ for $q>2$, and lies between $\tilde{O}(1/\epsilon^{2/(q-1)})$ and $\Omega(1/\epsilon^{\max\{1/(q-1),2}\})$ for $1<q<2$.
Moreover, complexities of quantum spectrum learning and rank testing problems have also been thoroughly studied in \cite{donnell2015spectrum} without limitations in allowed quantum operations.

\paragraph{Learning with restricted measurements and quantum replica advantage} Because implementing joint measurements across many replicas is experimentally challenging, recent works have explored learning properties under restricted measurement models. 
One prominent example is the classical shadow protocol~\cite{huangPredictingManyProperties2020}, which recovers quantum properties of a state from the outcomes of random 1-replica measurements, including $\tr(\rho^kO)$~\cite{hu2022logical,seif2023shadow}. 
Another example is the randomized measurement protocol, which is initially developed to estimate state moments $\tr(\rho^k)$ also with 1-replica measurements~\cite{van2012Measuring,elben2019toolbox,elben2023randomized,Brydges2019Probing}.
It reaches the provable optimal sample complexity in estimating state purity $\tr(\rho^2)$ among all 1-replica protocols~\cite{chenExponentialSeparationsLearning2022}.
This approach has also been applied to estimating $\tr(\rho^2O)$ \cite{duOptimalRandomizedMeasurements2025}. 
In \cite{pelecanosBeatingFullState2025}, another 1-replica protocol was proposed to estimate $\tr(\rho^k)$, using $\mathcal{O}(d^{3-2/k})$ copies to achieve constant multiplicative error.

However, it is already known that 1-replica protocols cannot efficiently estimate purity in the worst case~\cite{aharonovQuantumAlgorithmicMeasurement2022, chenExponentialSeparationsLearning2022,gongSampleComplexityPurity2024}, establishing an exponential separation between 1- and 2-replica protocols.
Another 1-versus-2 separation arises in shadow tomography~\cite{aaronsonShadowTomographyQuantum2018,chenExponentialSeparationsLearning2022,chenOptimalTradeoffsEstimating2024a,kingTriplyEfficientShadow}, where the goal is to estimate $\tr(P\rho)$ for all Pauli strings $P$. 
In \cite{chenHierarchyReplicaQuantum2021}, the authors construct a learning task that is hard for $k$-replica protocols but easy for $\mathrm{poly}(n,k)$-replica protocols. 

Even for 1-replica protocols, one already assumes that arbitrary quantum operations can be performed on a state $\rho$.
In the real world, quantum circuits are limited by factors including topology, noise, and physical characteristics.
Recently, researchers started to consider quantum learning with restricted circuit structures.
In \cite{huDemonstrationRobustEfficient2024,bertoni2022shallow,hu2023locallybias} the authors consider using shallow random circuits to perform classical shadow protocol, demonstrating intriguing scaling behavior in estimating certain nonlocal properties.
In \cite{Acharya2025pauli}, the authors give the optimal sample complexity of quantum tomography using Pauli measurements.
Compared with the known performance bound for 1-replica quantum tomography~\cite{chenWhenDoesAdaptivity2023}, they concluded that Pauli measurement is not the best method for quantum tomography.

\paragraph{Discussion about the concurrent work} As introduced above, there are two known sources of 1-vs-2 replica separations: the linearity of quantum mechanics and Heisenberg’s uncertainty principle. The first leads to the separation in estimating $\tr(\rho^2)$~\cite{aharonovQuantumAlgorithmicMeasurement2022, chenExponentialSeparationsLearning2022,gongSampleComplexityPurity2024}, while the second underlies the separation in Pauli shadow tomography~\cite{chenExponentialSeparationsLearning2022,chenOptimalTradeoffsEstimating2024a,kingTriplyEfficientShadow}. In the latter case, since Pauli strings $P$ are not pairwise commuting, one cannot estimate all $\tr(P\rho)$ simultaneously due to the uncertainty principle. However, $P\otimes P$ are pairwise commuting, which makes it possible to estimate all $\tr((P\otimes P)(\rho\otimes \rho))=\tr(P\rho)^2$ simultaneously.

Our work generalizes the first type of separation to all nonlinear properties $\tr(\rho^{k}O)$ with $\norm{O}_\infty=1$ and $\norm{O}_1=\Omega(d)$, establishing $(k-1)$-vs-$k$ replica separations for all positive integer $k$. Independently and concurrently, Nöller, Tran, Gachechiladze, and Kueng~\cite{Noeller2025hierarchy} extend the second type of separation: by generalizing Pauli tomography on qubits to Weyl–Heisenberg tomography on qudits, they obtained $(k-1)$-vs-$k$ replica separations for all prime $k$, thereby establishing another infinite replica hierarchy. 

\section{Preliminaries}
We use $\Id_D$ to denote the $D\times D$ identity operator and $\H_D$ to denote the $D$-dimensional Hilbert space. For nonnegative integers $a, b$, write $a^{\uparrow b}\coloneqq a(a+1)\cdots (a+b-1)$ and $a^{\downarrow b}\coloneqq a(a-1)\cdots (a-b+1)$, with the convention that $a^{\uparrow 0}=a^{\downarrow 0}=1$. Denote $[N]\coloneqq \{1, 2, \cdots, N\}$. For a pure state $\ket{\psi}$, we use $\psi$ to denote the density matrix $\ketbra{\psi}$. 

\subsection{Norms}
For a Hermitian matrix $A$, we use $\norm{A}_1, \norm{A}_F, \norm{A}_{\infty}$ to denote its trace norm, Frobenius norm, and operator norm, respectively. For a quantum state $\ket{\psi}$, we use $\norm{\ket{\psi}}_2$ to denote its $\ell_2$ norm.
The trace distance between two states is defined as $d_{\tr}(\rho, \sigma)\coloneqq \frac12\norm{\rho-\sigma}_1$.
\Cref{lem: norm facts} introduces some useful facts about norms.
\begin{lemma}\label{lem: norm facts}
    Let $A, B$ be Hermitian matrices and $\ket{\psi}, \ket{\phi}$ be two pure states. Then
    \begin{enumerate}[label=\alph*)]
        \item $\abs{\tr(AB)}\leq \norm{A}_1\norm{B}_{\infty}$.
        \item $\norm{A\otimes B}_1 = \norm{A}_1\norm{B}_1$.
        \item $\norm{\psi-\phi}_1\leq 2\norm{\ket{\psi}-\ket{\phi}}_2$.
        \item $\norm{\psi^{\otimes k}-\phi^{\otimes k}}_1\leq k\norm{\psi-\phi}_1$.
    \end{enumerate}
\end{lemma}

The following inequality is by Mirsky (see \cite[Corollary 7.4.9.3]{hornMatrixAnalysis1985}).

\begin{lemma}[Mirsky's inequality]
    For two quantum states $\sigma, \rho$, let $\bp = (p_1,\cdots, p_d)$ and $\bq=(q_1,\cdots, q_d)$ be spectra of them arranged in nonincreasing order. Then 
    \begin{equation}
        \norm{\sigma-\rho}_1\ge \sum_{i=1}^d \abs{p_i-q_i}.
    \end{equation}
\end{lemma}

\subsection{$k$-replica protocols}\label{sec: k replica}
To learn the desired properties of an unknown quantum state $\rho$, we need to perform measurements to the state. General quantum measurements are represented as positive operator-valued measures (POVMs).

\begin{definition}[POVMs]
    A $D$-dimensional positive operator-valued measure (POVM) is a set of positive-semidefinite matrices $\{F_s\in \mathbb{C}^{D\times D}\}_s$ that satisfies $\sum_s F_s=\Id_D$. When we measure a $D$-dimensional state $\sigma$, we obtain a classical outcome $s$ with probability $\tr(F_s\sigma)$.

    A $k$-replica joint measurement on $k$ copies of a $D$-dimensional state $\sigma$ is a $D^k$-dimensional POVM on the state $\sigma^{\otimes k}$. 
\end{definition}

A $k$-replica protocol is a learning protocol that can only perform joint measurements on at most $k$ copies of the unknown states at a time. We give a formal definition as follows.
\begin{definition}[$k$-replica protocols]\label{def: k-replica protocols}
    Given copy access to a $D$-dimensional unknown quantum state $\sigma$, a $k$-replica protocol with $T$ rounds works as follows: In each round of the protocol, it gets $k$ copies of $\sigma$ and measure $\sigma^{\otimes k}$ with a $D^k$-dimensional POVM $\{F_s\}_s$, obtaining some classical outcome $s$ with probability $\tr(F_s\sigma^{\otimes k})$. The choice of the POVM can depend on all previous measurement outcomes. After $T$ measurements, the algorithm predicts the desired properties of $\sigma$ according to all measurement outcomes. The sample complexity of the protocol is $kT$. 
\end{definition}

Alternatively, we can merge the $T$ rounds of measurements into one giant POVM on $kT$ qudits with a tree structure.

\begin{definition}[$k$-replica $T$-round POVMs]
    A $k$-replica $T$-round POVM with local dimension $D$ is a $D^{kT}$-dimensional POVM with the form $\{F_{s_1}\otimes F_{s_1,s_2}\otimes\cdots\otimes F_{s_1,s_2,\cdots, s_T}\}_{s_1,s_2,\cdots, s_T}$, where for each $t\leq T$ and $s_1,\cdots, s_{t-1}$, $\{F_{s_1,s_2,\cdots, s_{t-1}, s_t}\}_{s_t}$ is a $D^k$-dimensional POVM.
\end{definition}

\begin{fact}
    We can redefine a $k$-replica protocol with $T$ rounds as follows: It performs a $k$-replica $T$-round POVM to $kT$ copies of the unknown state, and predicts desired properties according to classical outcomes $s_1,\cdots, s_T$.
\end{fact}

In this paper, we do not utilize the tree structure of $k$-replica protocols. Instead, we only need the fact that every POVM element has a product form. 

\begin{definition}[$(D,T)$-product POVMs]
    A $(D, T)$-product POVM is a $D^{T}$-dimensional POVM $\{F_s\}_s$ such that each $F_s$ is a tensor product of $T$ $D$-dimensional operators, namely, $F_s\coloneqq F_{s,1}\otimes F_{s,2}\otimes\cdots \otimes F_{s,T}$. Denote the set of $(D,T)$-product POVMs by $\cM_{D,T}$. 
\end{definition}

\begin{fact}\label{fact: k t product}
    A $k$-replica $T$-round POVM with local dimension $D$ is a $(D^k,T)$-product POVM.
\end{fact}

The main tool to prove lower bounds is Le Cam's two-point method.

\begin{task}[$\cE_1$-versus-$\cE_2$ distinguishing problem]
    Given copy access to an unknown quantum state $\sigma$, promised to be drawn from one of two ensembles $\cE_1$ and $\cE_2$, the goal is to determine which ensemble $\sigma$ is drawn from.
\end{task}

\begin{lemma}[Le Cam's two-point method]\label{lem: Le Cam}
    Let $M=\{F_s\}_s$ be a $D$-dimensional POVM and $\cE_1,\cE_2$ be two $D$-dimensional ensembles. The success probability of distinguishing $\cE_1,\cE_2$ by performing $M$ once is at most
    \begin{equation}
        d_M(\E_{\sigma\leftarrow \cE_1}[\sigma],\E_{\sigma\leftarrow \cE_2}[\sigma])\coloneqq \frac{1}{2}\sum_s \abs{\tr(F_s\E_{\sigma\leftarrow \cE_1}[\sigma])-\tr(F_s\E_{\sigma\leftarrow \cE_2}[\sigma])}.
    \end{equation}
    Here $d_M(\sigma_1,\sigma_2)$ is the total variation distance between the outcome distributions obtained by measuring the two states using $M$.
\end{lemma}

Therefore, the problem of proving lower bounds for sample complexity reduces to upper bounding the total variation distance. We recapitulate the reduction in the following lemma.

\begin{lemma}\label{lem: reduction to TV}
    Let $T_0\in\N^+$ and $c_0\in \R^+$. For two $D$-dimensional ensembles $\cE_1, \cE_2$, if we can prove that 
    \begin{equation}
        \max_{M\in \cM_{D, T}}d_M(\E_{\sigma\leftarrow \cE_1}[\sigma^{\otimes T}], \E_{\sigma\leftarrow \cE_2}[\sigma^{\otimes T}])\leq c_0,~\forall T<T_0,
    \end{equation}
    then the sample complexity of distinguishing $\cE_1$ and $\cE_2$ with success probability $c_0$ using 1-replica protocols is at least $T_0$.

    In particular, for two $D$-dimensional ensembles $\cE_1', \cE_2'$, if we can prove that
    \begin{equation}
        \max_{M\in \cM_{D^k, T}}d_M(\E_{\sigma\leftarrow \cE_1'}[\sigma^{\otimes kT}], \E_{\sigma\leftarrow \cE_2'}[\sigma^{\otimes kT}])\leq c_0,~\forall T<T_0.
    \end{equation}
    then the sample complexity of distinguishing $\cE_1'$ and $\cE_2'$ with success probability $c_0$ using $k$-replica protocols is at least $kT_0$.
\end{lemma}

One useful trick to upper bound the total variation distance is through the one-sided likelihood ratio.
\begin{lemma}[One-sided likelihood ratio bound suffices, adapted from {\cite[Lemma 5.4]{chenExponentialSeparationsLearning2022}}]\label{lem: one-sided likelihood ratio}
    Let $\delta>0$ and let $\{p_i\}_i, \{q_i\}_i$ be two probability distributions such that the one-sided likelihood ratio, $p_i/q_i$, is lower bounded by $1-\delta$ for every $i$ (this condition automatically holds if $q_i=0$). Then the total variation distance, $d(\{p_i\}_i,\{q_i\}_i)\coloneqq \frac12 \sum_i \abs{p_i-q_i}$, is upper bounded by $\delta$. 
\end{lemma}

\subsection{Haar measure and permutation operators}\label{sec: Haar measure}
Random instances are common tools to prove lower bounds on the sample complexity of learning problems. In quantum information, the uniform distributions over quantum states and unitaries are captured by Haar measure, denoted by $\psi\leftarrow \mu_\Haar(D)$ and $U\leftarrow \mu_\Haar(D\times D)$, respectively. The moments of Haar random states can be nicely represented using permutation operators.

\begin{definition}[Permutation operators]
    For $N\in \N^+$, let $S_N$ be the symmetric group on $N$ elements. For each permutation $\pi\in S_N$, we can define a corresponding permutation operator on $N$ qudits with local dimension $D$:
    \begin{equation}
        \pi^{(D)} \ket{\psi_1}\otimes\cdots\otimes \ket{\psi_N}=\ket{\psi_{\pi^{-1}(1)}}\otimes \cdots\otimes\ket{\psi_{\pi^{-1}(N)}}.
    \end{equation}
    Denote the symmetry operator $S_N^{(D)}\coloneqq \sum_{\pi\in S_N}\pi^{(D)}$ and the normalized symmetry operator $\tilde{S}_N^{(D)}\coloneqq \frac{1}{D^{\uparrow N}}S_N^{(D)}$. We sometimes omit the superscriptions $(D)$ in $S_N^{(D)}$ and $\tilde{S}_N^{(D)}$. 
\end{definition}
\begin{lemma}[Haar integrals for pure states]\label{lem: Haar integrals}
    $\E_{\psi\leftarrow \mu_\Haar(D)}\psi^{\otimes k} = \frac{1}{D(D+1)\cdots (D+k-1)}\sum_{\pi\in S_k}\pi = \tilde{S}_k$. 
\end{lemma}

Another desirable property of Haar measure is the strong concentration of measure given by Levy's lemma. We say a function $f: \H_d \to \R$ is $L$-Lipschitz if 
\begin{equation}
    \abs{f(\psi)-f(\phi)}\leq L\norm{\psi-\phi}_2.
\end{equation}
\begin{lemma}[Levy's lemma]\label{lem: Levy's lemma}
    Let $f: \H_d\to \R$ be an $L$-Lipschitz function. For any $\delta>0$,
    \begin{equation}
        \Pr_{\psi\leftarrow \mu_\Haar(d)}\bracks*{\abs{f(\psi)-\E_{\phi\leftarrow \mu_\Haar(d)}[f(\phi)]}>\delta} \leq 2\exp(-\frac{2d\delta^2}{9\pi^3 L^2}).
    \end{equation}
\end{lemma}
There is a stronger concentration for inner products
\begin{lemma}[See {\cite[Proposition 4.8]{kuengACM270Quantum2022}} or {\cite[Example 55]{meleIntroductionHaarMeasure2024}}]\label{lem: haar concentration inner product}
    Let $\psi$ be a Haar random state and $\phi$ be a fixed state.
    \begin{equation}
        \Pr_{\psi}[\abs{\norm{\braket{\psi}{\phi}}^2}\ge \delta]\leq 2\exp(-\frac{d}{2}\delta).
    \end{equation}
\end{lemma}

\paragraph{$k$-body component}
For a $k$-qudit observable $O$, the moment operator $\E_{U\leftarrow \mu_\Haar}[U^{\otimes k}O (U^\dagger)^{\otimes k}]$ is a linear combination of permutation operators, expressed as
\begin{equation}
    \E_{U\leftarrow \mu_\Haar}[U^{\otimes k}O (U^\dagger)^{\otimes k}] = \sum_{\pi\in S_k}c_\pi(O)\pi. 
\end{equation}
Let $S_k^{\text{circ}}$ be the set of circular permutations. We say $c_\pi(O)\pi$ is a $k$-body component of $O$ if $\pi\in S_k^{\text{circ}}$ as it cannot be decomposed to a tensor product of smaller operators. The $k$-body weight of $O$ is defined as $w(O)\coloneqq \abs{\sum_{\pi\in S_k^{\text{circ}}}c_\pi(O)}$. 

\subsection{Subroutines}\label{sec: subroutines}
\paragraph{Generalized swap test} Given a state $\sigma$ and a unitary $U$, the generalized swap test (adapted from \cite{ekertDirectEstimationsLinear2002}) goes as follows: Apply a controlled-$U$ to $\ket{+}\otimes \sigma$ and measure the first qubit in the $X$ basis. With probability $\frac{1}{4}\tr((\Id\pm U)\sigma(\Id\pm U^\dagger))$, the outcome is $\ket{\pm}$ and the (unnormalized) post-measurement state is $(\Id\pm U)\sigma(\Id\pm U^\dagger)$.

When $\sigma=\rho^{\otimes k+1}$ for a single-qudit state $\rho$ and $U=\pi$ is a $(k+1)$-qudit cyclic permutation operator. The generalized swap test outputs $\ket{\pm}$ with probability $\frac14\tr((\Id\pm \pi)\rho^{\otimes k+1}(\Id\pm \pi))=\frac{1\pm \tr(\rho^{k+1})}{2}$. Tracing out the last $k$ qudits, the normalized state in the remaining register is $\frac{\rho\pm \rho^{k+1}}{1\pm \tr(\rho^{k+1})}$ upon the outcome $\ket{\pm}$. 
\paragraph{Direct estimation of nonlinear observables} Given a $d$-dimensional state $\rho$ and a $d$-dimensional observable $O$, estimating $\tr(\rho^{k+1} O)$ is easy if we can perform $(k+1)$-replica joint measurements. Indeed, we can write $\tr(\rho^{k+1}O)$ as the linear function of $\rho^{\otimes k+1}$: $\tr(\rho^{\otimes k+1}\cdot\frac12[(O\otimes \Id_{d^k})\pi + \pi^\dagger (O\otimes \Id_{d^k})])$ for any $(k+1)$-body cyclic permutation. Thus, it suffices to measure $\rho^{\otimes k+1}$ in the eigenbasis of the Hermitian operator $\frac12[(O\otimes \Id_{d^k})\pi + \pi^\dagger (O\otimes \Id_{d^k})]$. The overall sample complexity is $O(k/\epsilon^2)$, where $\epsilon$ is the estimation error. 

However, this method requires the diagonalization of $\frac12[(O\otimes \Id_{d^k})\pi + \pi^\dagger (O\otimes \Id_{d^k})]$. A more clever way is to use a variety of generalized swap test. Apply a controlled-$\pi$ to $\ket{+}\otimes \rho^{\otimes k+1}$. We can verify that the expectation of $X\otimes O\otimes \Id_{d}^{\otimes k}$ under the resulting state is equal to $\tr(\rho^{k+1}O)$, thus it can be estimated directly using $O(k/\epsilon^2)$ samples. 

\paragraph{Distributed estimation of $\tr(\rho\sigma O)+\tr(\sigma\rho O)$} Given two states $\rho$ and $\sigma$, the goal of inner product estimation is to estimate $\tr(\rho\sigma)$. 
In \cite{anshuDistributedQuantumInner2022}, the authors characterize the sample complexity of this task in the distributed setting, which is equivalent to the 1-replica protocol defined in this work.
In \cite{duOptimalRandomizedMeasurements2025}, the authors consider the estimation of $\tr(\rho^2O)$ with 1-replica protocol, finding a similar sampling overhead with the purity estimation.
Their techniques actually can be easily combined and applied to the estimation of $\tr(\rho\sigma O)+\tr(\sigma\rho O)$ with 1-replica protocols.
\begin{lemma}[Inner product estimation \cite{anshuDistributedQuantumInner2022,duOptimalRandomizedMeasurements2025}]\label{lem: inner product estimation}
There exists an algorithm for estimating $\tr(\rho\sigma O)+\tr(\sigma\rho O)$ with arbitrary observable $\norm{O}_\infty=1$ with sample complexity $\tilde{\mathcal{O}}(\max\{1/\epsilon^2, \sqrt{d}/\epsilon\})$. The algorithm performs 1-replica measurement to each copy of $\rho$ and $\sigma$ separately and outputs the estimation based on all measurement outcomes. Moreover, the algorithm is non-adaptive, meaning that all measurements are pre-determined and are independent of intermediate measurement outcomes.
\end{lemma}
\noindent This lemma is crucial for proving the upper bound result in Theorem~\ref{thm: upper bound}, where we mainly consider the case $\rho$ and $\sigma$ commute with each other.
Therefore, we can simplify $\tr(\rho\sigma O)+\tr(\sigma\rho O)$ into $\tr(\rho\sigma O)$.

\section{Indistinguishability principle for Haar-assembled ensembles}\label{sec: indistinguishability principle}
In this section, we prove the hardness of distinguishing a Haar-assembled ensemble from its average. We restate the definition of Haar-assembled ensembles here for convenience. We slightly generalize the definition in \Cref{def: Haar assembled ensembles} by allowing Haar random states of different dimensions. This generalization is not needed for our main results but may be useful for future applications.
\begin{definition}[Haar-assembled ensembles]\label{def: haar assembled ensembles restated}
    For $m,D\in \N^+$, a $D$-dimensional Haar-assembled ensemble $\cE_\Haar$ assembled by $m$ Haar-random states $\psi_1,\cdots, \psi_m$ with local dimensions $d_1,\cdots, d_m$ is defined as
    \begin{equation}
        \cE_\Haar \coloneqq \left\{\sum_{j=1}^J p_j U_j(\psi_1^{\otimes a_{j1}}\otimes \cdots \otimes \psi_m^{\otimes a_{jm}}\otimes \tau_j)U_j^\dagger\right\}_{\psi_r \leftarrow \mu_\Haar(d_r),r\in[m]}\,,\label{equ: Haar assembled ensembles restated}
    \end{equation}
    where $J\in \N^+$, $\{p_j\}_{j=1}^J$ is a probability distribution. For each $j\in [J]$ and $r\in [m]$, $a_{jr}\in \N$, $U_j$ is a fixed unitary, and $\tau_j$ is a fixed state such that the dimension of each term in the summation is $D$.

    For $r\in [m]$, define $a_{r}\coloneqq \sum_{j=1}^J p_j a_{jr}$ as the average occurrence of $\psi_r$, and $a'_{r}\coloneqq \sum_{j=1}^J p_j a_{jr}^2$ as the average squared occurrence of $\psi_r$. Denote $d_{\min} \coloneqq \min_{r\in[m]} d_r$ and $a_{\max}\coloneqq \max_{r\in[m]} a_r$.
\end{definition}

By \Cref{lem: Haar integrals}, the average of the Haar-assembled ensemble is given by
\begin{equation}
    \E_{\sigma\leftarrow \cE_\Haar}[\sigma] \coloneqq \sum_{j=1}^J p_j U_j(\tilde{S}_{a_{j1}}^{(d_1)}\otimes \cdots \otimes \tilde{S}_{a_{jm}}^{(d_m)}\otimes \tau_j)U_j^\dagger\,.
\end{equation}
We will prove indistinguishability from two aspects:
\begin{itemize}
    \item In \Cref{thm: concentration of Haar-assembled ensemble}, we prove the concentration of $\tr(O\sigma)$ for $\sigma\leftarrow \cE_\Haar$. Therefore, it is hard to distinguish a random instance from its mean by estimating the expectation of a single observable $O$.
    \item In \Cref{thm: Haar-assembled ensemble vs mean}, we show that any 1-replica protocol cannot distinguish $\sigma\leftarrow \cE_\Haar$ from $\E_{\sigma'\leftarrow \cE_\Haar}[\sigma']$ efficiently with nontrivial advantage. We call this the indistinguishability principle for Haar-assembled ensembles.
\end{itemize}

\subsection{Concentration of Haar-assembled ensembles}
We first prove the concentration of Haar-assembled ensembles. 

\begin{theorem}\label{thm: concentration of Haar-assembled ensemble}
    Let $\cE_\Haar$ be a Haar-assembled ensemble defined in \eqref{equ: Haar assembled ensembles restated} and $O$ be a $D$-dimensional observable. For all $\delta>0$,
    \begin{equation}
        \Pr_{\sigma\leftarrow \cE_\Haar}\bracks*{\abs{\tr(O\sigma)-\tr(O \E_{\sigma'\leftarrow \cE_\Haar}[\sigma'])} > \delta} \leq 2\sum_{r=1}^m\exp(-\frac{d_r\delta^2}{18\pi^3 m^2 a_r^2 \norm{O}_{\infty}^2}).
    \end{equation}
\end{theorem}
We apply Levy's lemma recursively to prove the theorem. The following lemma establishes the Lipschitzness of $\tr(O\sigma)$ as a function of each Haar-random state $\psi_i$.
\begin{lemma}
    Let $O$ be a $D$-dimensional observable and $\sigma = \sum_{j=1}^J p_j U_j(\psi_1^{\otimes a_{j1}}\otimes \cdots \otimes \psi_m^{\otimes a_{jm}}\otimes \tau_j)U_j^\dagger$ as in~\eqref{equ: Haar assembled ensembles restated}. Then for each $r\in [m]$, $\tr(O\sigma)$ is $2a_r\norm{O}_{\infty}$-Lipschitz as a function of $\psi_r$, while fixing all other $\psi_i$s.
\end{lemma}
\begin{proof}
    We write $\sigma=\sigma(\psi_r)$ to emphasize the dependence on $\psi_r$. For any two $d_r$-dimensional states $\phi, \phi'$,
    \begin{align}
        &\abs{\tr(O\sigma(\phi))-\tr(O\sigma(\phi'))}\nonumber\\
        \leq & \norm{O}_{\infty}\norm{\sigma(\phi)-\sigma(\phi')}_1\tag*{(\Cref{lem: norm facts}(a))}\\
        \leq& \norm{O}_{\infty}\sum_{j=1}^J p_j \norm{(\phi^{\otimes a_{jr}}-(\phi')^{\otimes a_{jr}})\otimes \tau_j\otimes \bigotimes_{i\neq r}\psi_i^{\otimes a_{ji}}}_1\tag*{(triangle inequality)}\\
        =& \norm{O}_{\infty}\sum_{j=1}^J p_j \norm{\phi^{\otimes a_{jr}}-(\phi')^{\otimes a_{jr}}}_1\tag*{(\Cref{lem: norm facts}(b))}\\
        \leq & \norm{O}_{\infty}\sum_{j=1}^J p_j a_{jr}\norm{\phi-\phi'}_1\tag*{(\Cref{lem: norm facts}(d))}\\
        = & a_r\norm{O}_{\infty}\norm{\phi-\phi'}_1\tag{$a_r=\sum_{j=1}^Jp_ja_{jr}$}\\
        = & 2a_r\norm{O}_{\infty}\norm{\ket{\phi}-\ket{\phi'}}_2.\tag*{(\Cref{lem: norm facts}(c))}
    \end{align}
\end{proof}

Now we are ready to prove \Cref{thm: concentration of Haar-assembled ensemble}.
\begin{proof}[Proof of \Cref{thm: concentration of Haar-assembled ensemble}]
    The randomness of $\sigma\leftarrow \cE_\Haar$ comes from $m$ independent Haar random state $\psi_1, \cdots, \psi_m$. We write $\tr(O\sigma)$ as $f(\psi_1, \cdots, \psi_m)$ to emphasize this dependence. By telescoping,
    \begin{align}
        &\abs{f(\psi_1, \cdots, \psi_m)-\E_{\psi_1',\cdots, \psi_m'}[f(\psi_1',\cdots, \psi_m')]} \nonumber\\
        \leq&  \sum_{r=1}^m \abs{\E_{\psi_1', \cdots, \psi_{r-1}'}f(\psi_1', \cdots, \psi_{r-1}', \psi_r, \cdots, \psi_m) - \E_{\psi_1', \cdots, \psi_{r}'}f(\psi_1', \cdots, \psi_{r}', \psi_{r+1}, \cdots, \psi_m)}. \label{equ: proof of concentration}
    \end{align}
    Regard $\E_{\psi_1', \cdots, \psi_{r-1}'}f(\psi_1', \cdots, \psi_{r-1}', \psi_r, \cdots, \psi_m)$ as a function of $\psi_r$ for fixed $\psi_{r+1}, \cdots, \psi_m$. By the previous lemma, it is $2a_r\norm{O}_{\infty}$-Lipschitz. Applying Levy's lemma (\Cref{lem: Levy's lemma}), the $r$-th term in \eqref{equ: proof of concentration} is at most $\delta/m$ with probability at least $1-2\exp(-\frac{d_r\delta^2}{18\pi^3 m^2 a_r^2 \norm{O}_{\infty}^2})$. 
    By union bound, with probability at least $1-2\sum_{r=1}^m\exp(-\frac{d_r\delta^2}{18\pi^3 m^2 a_r^2 \norm{O}_{\infty}^2})$, every term in \eqref{equ: proof of concentration} is at most $\delta/m$. Therefore, with the same probability, $\abs{\tr(O\sigma)-\tr(O \E_{\sigma'\leftarrow \cE_\Haar}[\sigma'])} \leq \delta$. 
\end{proof}

\subsection{Haar-assembled ensemble versus its average}
We now show the hardness of distinguishing in a stronger sense. Suppose we have copy access to $\sigma$ that is either drawn from a Haar-assembled ensemble $\cE_\Haar$ or is equal to the averaged state $\E_{\sigma'\leftarrow \cE_\Haar}[\sigma']$. In this section, we lower bound the sample complexity of distinguishing the two cases using $1$-replica protocols.

By \Cref{lem: reduction to TV}, we only need to upper bound
\begin{equation}
    \max_{M\in \cM_{D,T}} d_{M}(\E_{\sigma\leftarrow \cE_\Haar}[\sigma^{\otimes T}], \E_{\sigma\leftarrow \cE_\Haar}[\sigma]^{\otimes T}).
\end{equation}

\begin{theorem}[Indistinguishability principle for Haar-assembled ensembles]\label{thm: Haar-assembled ensemble vs mean}
    Let $T\in\N^*$ and $\cE_\Haar$ be a $D$-dimensional Haar-assembled ensemble defined in \eqref{equ: Haar assembled ensembles restated}:
    \begin{equation}
        \cE_\Haar \coloneqq \left\{\sum_{j=1}^J p_j U_j(\psi_1^{\otimes a_{j1}}\otimes \cdots \otimes \psi_m^{\otimes a_{jm}}\otimes \tau_j)U_j^\dagger\right\}_{\psi_r \leftarrow \mu_\Haar(d_r),r\in[m]}\,.
    \end{equation}
    Any $1$-replica $T$-round protocol can only distinguish between $\cE_\Haar$ and $\E_{\sigma\leftarrow \cE_\Haar}[\sigma]$ with success probability at most
    \begin{equation}\label{equ: haar assembled ensemble vs mean}
        \max_{M\in \cM_{D,T}} d_{M}(\E_{\sigma\leftarrow \cE_\Haar}[\sigma^{\otimes T}], \E_{\sigma\leftarrow \cE_\Haar}[\sigma]^{\otimes T}) \leq  T(T-1)\sum_{r=1}^m \frac{a_r^2}{d_r}+T\sum_{r=1}^m\frac{a_r'}{d}.
    \end{equation}
    As a consequence, $T=\Omega(\sqrt{d_{\min}/(ma_{\max}^2)})$ is necessary to achieve a constant success probability.
\end{theorem}
\begin{proof}
    Write $\sigma_j=U_j(\psi_1^{\otimes a_{j1}}\otimes \cdots \otimes \psi_m^{\otimes a_{jm}}\otimes \tau_j)U_j^\dagger$. Then $\sigma = \sum_{j=1}^J p_j \sigma_j$. Fix a $M\in \cM_{D, T}$. We will omit the subscription $\psi_r\leftarrow \psi_\Haar(d_r)$ in all expectations below. Expanding 
    \begin{equation}
        \E[\sigma^{\otimes T}] = \sum_{j_1,\cdots, j_T\in [J]} \parens{\prod_{t=1}^T p_{j_t}} \E[\otimes_{t=1}^T \sigma_{j_t}],
    \end{equation}
    and 
    \begin{equation}
        \E[\sigma]^{\otimes T} = \sum_{j_1,\cdots, j_T\in [J]} \parens{\prod_{t=1}^T p_{j_t}} \otimes_{t=1}^T \E[\sigma_{j_t}]
    \end{equation}
    in $d_M(\E[\sigma^{\otimes T}], \E[\sigma]^{\otimes T})$, and applying triangle inequality, we have
    \begin{align}
        d_M(\E[\sigma^{\otimes T}], \E[\sigma]^{\otimes T})\leq
        \sum_{j_1,\cdots, j_T\in [J]} \parens{\prod_{t=1}^T p_{j_t}} d_M(\E[\otimes_{t=1}^T \sigma_{j_t}], \otimes_{t=1}^T \E[\sigma_{j_t}]). \label{equ: main 1}
    \end{align}
    Fix a term $j_1, \cdots, j_T\in [J]$. For $r\in [m]$, let $A_r=\sum_{t=1}^T a_{j_t,r}$ be the total number of occurrences of $\psi_r$ in $\otimes_{t=1}^T \sigma_{j_t}$. We will upper bound $d_M(\E[\otimes_{t=1}^T \sigma_{j_t}], \otimes_{t=1}^T \E[\sigma_{j_t}])$ by
    \begin{equation}
        d_M(\E[\otimes_{t=1}^T \sigma_{j_t}], \otimes_{t=1}^T \E[\sigma_{j_t}]) \leq 1-\exp(-\sum_{r=1}^m \frac{A_r^2}{d_r}) \leq \sum_{r=1}^m \frac{A_r^2}{d_r}. \label{equ: main 2}
    \end{equation}
    Let us first assume \eqref{equ: main 2} and finish the proof. Plugging \eqref{equ: main 2} into \eqref{equ: main 1}, we obtain
    \begin{equation}
        d_M(\E[\sigma^{\otimes T}], \E[\sigma]^{\otimes T}) \leq \sum_{j_1,\cdots, j_T\in [J]} \parens*{\prod_{t=1}^T p_{j_t}} \sum_{r=1}^m \frac{1}{d_r}\parens*{\sum_{t=1}^T a_{j_t, r}}^2.
    \end{equation}
    Let $b_j = \abs{\{t: j_t=j\}}$. The right-hand side can be rearranged as
    \begin{align}
        &d_M(\E[\sigma^{\otimes T}], \E[\sigma]^{\otimes T}) \nonumber\\
        \leq& \sum_{b_1+\cdots + b_J=T}\binom{T}{b_1,b_2,\cdots, b_J}\parens*{\prod_{j=1}^J p_j^{b_j}} \sum_{r=1}^m \frac{1}{d_r}\parens*{\sum_{j=1}^J b_j a_{jr}}^2\nonumber\\
        =& T(T-1)\sum_{r=1}^m \frac{1}{d_r}\parens*{\sum_{j=1}^J p_ja_{jr}}^2+T\sum_{r=1}^m\frac{1}{d_r}\sum_{j=1}^Jp_ja_{jr}^2\,. \tag*{(\Cref{cor: expectation under multinomial distribution quadratic})}
    \end{align}
    This gives \eqref{equ: haar assembled ensemble vs mean}. Since $d_r\ge d_{\min}$ and $a_{jr}\leq a_{\max}$, we have $a_{r}\leq a_{\max}$ and $a'_{r}\leq a_{\max}^2$, so the right-hand side of \eqref{equ: haar assembled ensemble vs mean} is at most $T^2a_{\max}^2m/d_{\min}$. Therefore, $T=\Omega(\sqrt{d_{\min}/(ma_{\max}^2)})$ is necessary to achieve a constant success probability. So the theorem holds once we assume \eqref{equ: main 2}.

    Now we prove \eqref{equ: main 2}. Write $M=\{F_s\}_s$ By definition, $d_M(\E[\otimes_{t=1}^T \sigma_{j_t}], \otimes_{t=1}^T \E[\sigma_{j_t}])$ is the total variation distance between $\{\tr(F_s\E[\otimes_{t=1}^T \sigma_{j_t}])\}_{s}$ and $\{\tr(F_s\otimes_{t=1}^T \E[\sigma_{j_t}])\}_{s}$. \Cref{lem: one-sided likelihood ratio} allows us to upper bound this total variation distance by lower bounding the one-sided likelihood ratio $\tr(F_s\E[\otimes_{t=1}^T \sigma_{j_t}])/\tr(F_s\otimes_{t=1}^T \E[\sigma_{j_t}])$ for every $s$. By definition of $\cM_{D, T}$, each $F_s$ has a product form. So it suffices to prove that for any $F=F_1\otimes \cdots \otimes F_T$ where each $F_t$ is a $D$-dimensional positive semidefinite operator,
    \begin{equation}
        \frac{\tr(F\E[\otimes_{t=1}^T \sigma_{j_t}])}{\tr(F\otimes_{t=1}^T \E[\sigma_{j_t}])} \geq \exp(-\sum_{r=1}^m \frac{A_r^2}{d_r}). \label{equ: main 3}
    \end{equation}

    We simplify \eqref{equ: main 3} by absorbing $U_{j_t}$ and $\tau_{j_t}$ into $F_t$. Specifically, for $t\in [T]$, define $$G_t\coloneqq \tr_{\tau_{j_t}}((U^\dagger_{j_t}F_t U_{j_t})\cdot (\Id^{-\tau_{j_t}}\otimes \tau_{j_t})),$$
    where by $\tr_{\tau_{j_t}}$ we mean tracing out the subsystem occupied by $\tau_{j_t}$ and by $\Id^{-\tau_{j_t}}$ we mean the identity operator acting on the complement of this subsystem. 
    We can verify that 
    \begin{equation}
        \tr(F_t\sigma_{j_t})=\tr(F_t U_{j_t}(\psi_1^{\otimes a_{j_t,1}}\otimes \cdots \otimes \psi_m^{\otimes a_{j_t,m}}\otimes \tau_{j_t})U_{j_t}^\dagger)=\tr(G_t \psi_1^{\otimes a_{j_t,1}}\otimes \cdots \otimes \psi_m^{\otimes a_{j_t,m}}).
    \end{equation}

    Define $G\coloneqq G_1\otimes \cdots \otimes G_T$. The left-hand side of \eqref{equ: main 3} becomes
    \begin{equation}
        \frac{\tr(F\E[\otimes_{t=1}^T \sigma_{j_t}])}{\tr(F\otimes_{t=1}^T \E[\sigma_{j_t}])} = \frac{\tr(G\E[\otimes_{t=1}^T (\psi_1^{\otimes a_{j_t,1}}\otimes \cdots \otimes \psi_m^{\otimes a_{j_t,m}})])}{\tr(G\otimes_{t=1}^T \E[\psi_1^{\otimes a_{j_t,1}}\otimes \cdots \otimes \psi_m^{\otimes a_{j_t,m}}])}. \label{equ: main 4 prep}
    \end{equation}
    To express the Haar integral, we need to mark positions of operators in a clear way. For a $t\in [T]$, we regard $\psi_1^{\otimes a_{j_t, 1}}\otimes \cdots \otimes \psi_m^{\otimes a_{j_t, m}}$ as a state on $\sum_{r=1}^m a_{j_t, r}$ qudits, where $\psi_r^{\otimes a_{j_t, r}}$ lives on $a_{j_t,r}$ qudits among them with local dimension $d_r$. Denote the set of these qudits by $A_{t, r}$. For a set $A$ of qudits (with the same local dimension), we use $S^A$ to denote the symmetry operator on qudits in $A$. With these notations, we have
    \begin{equation}
        \bigotimes_{t=1}^T (\psi_1^{\otimes a_{j_t,1}}\otimes \cdots \otimes \psi_m^{\otimes a_{j_t,m}}) = \bigotimes_{t=1}^T(\psi_1^{\otimes A_{t, 1}}\otimes \cdots \otimes \psi_m^{\otimes A_{t, m}}) = \bigotimes_{r=1}^m \psi_r^{\otimes \bigcup_{t=1}^T A_{t, r}}.
    \end{equation}
    Now we can express the likelihood ratio in \eqref{equ: main 4 prep} as 
    \begin{align}
        \frac{\tr(F\E[\otimes_{t=1}^T \sigma_{j_t}])}{\tr(F\otimes_{t=1}^T \E[\sigma_{j_t}])} &= \frac{\tr(G\otimes_{r=1}^m \E[\psi_r^{\otimes\bigcup_{t=1}^T A_{t,r}}])}{\tr(G\otimes_{t=1}^T\otimes_{r=1}^m \E[\psi_r^{\otimes A_{t,r}}])}\nonumber\\
        &= \frac{\prod_{t=1}^T\prod_{r=1}^m d_r^{\uparrow a_{j_t,r}}}{\prod_{r=1}^m d_r^{\uparrow \sum_{t=1}^T a_{j_t, r}}} \cdot \frac{\tr(G\otimes_{r=1}^m S^{\bigcup_{t=1}^T A_{t,r}})}{\tr(G\otimes_{t=1}^T\otimes_{r=1}^m S^{A_{t,r}})}\tag*{(\Cref{lem: Haar integrals})}\\
        &\geq \frac{\prod_{t=1}^T\prod_{r=1}^m d_r^{\uparrow a_{j_t,r}}}{\prod_{r=1}^m d_r^{\uparrow \sum_{t=1}^T a_{j_t, r}}}\label{equ: main 4} \\
        &\geq \prod_{r=1}^m \exp(-\frac{1}{d_r}\parens*{\sum_{t}^T a_{j_t,r}}^2)\label{equ: main 5}\\
        &= \exp(-\sum_{r=1}^m \frac{A_r^2}{d_r}). \tag*{($A_r=\sum_{t=1}^T a_{j_t,r}$)}
    \end{align}
    Here \eqref{equ: main 4} is by the \Cref{lem: permutation inequality} below, and \eqref{equ: main 5} is from 
    \begin{align*}
        \frac{\prod_{t=1}^T d^{\uparrow y_t}}{d^{\uparrow y_1+\cdots y_T}}&=\prod_{t=1}^{T}\frac{d^{\uparrow y_t}}{(d+y_1+\cdots +y_{t-1})^{\uparrow y_t}}\\
        &\ge \prod_{t=1}^T \left(\frac{d}{d+y_1+\cdots + y_{t-1}}\right)^{y_t}\\
        &\ge \prod_{t=1}^T \exp(-\frac{y_t(y_1+\cdots y_{t-1})}{d})\\
        &= \exp(-\frac{1}{d}\sum_{t_1<t_2}y_{t_1}y_{t_2})\\
        &\ge \exp(-\frac{1}{d}\parens*{\sum_{t}y_t}^2).
    \end{align*}
    Now we have proved \eqref{equ: main 3}, and thus \eqref{equ: main 2} by \Cref{lem: one-sided likelihood ratio}. This completes the proof of the theorem.
\end{proof}

\begin{lemma}\label{lem: permutation inequality}
    Let $I$ be a set of qudits, $T, m\in \N^+$, and let $\{I_{t,r}\}_{t\in [T], r\in [m]}$ be a partition of $I$. Write $P_t=\bigcup_{r=1}^m I_{t,r}$ and $Q_r=\bigcup_{t=1}^T I_{t,r}$, so
    $\{P_1,\dots,P_T\}$ and $\{Q_1,\dots,Q_m\}$ are two partitions of $I$. Suppose $G$ is a positive semidefinite operator on $I$ that factorizes across the partition $\{P_t\}_{t\in [T]}$, i.e.,
    \begin{equation}
        G = G_1 \otimes \cdots \otimes G_T \quad\text{with } G_t \text{ acting on } P_t.
    \end{equation}
    Then 
    \begin{equation}
        \tr(G \otimes_{r=1}^m S^{Q_r}) \geq \tr(G \otimes_{t=1}^T \otimes_{r=1}^m S^{I_{tr}})=\prod_{t=1}^T \tr(G_t\otimes_{r=1}^m S^{I_{tr}})\,. \label{equ: permutation inequality}
    \end{equation}
    Here we do note require that all qudits have the same local dimension. We only require that qudits in each $Q_r$ have the same local dimension.
\end{lemma}
When $I=[T]$, $m=1$, $I_{t, 1}=\{t\}$, this lemma reduces to $\tr((G_1\otimes\cdots \otimes G_T) S_T)\ge \tr(G_1)\times \cdots \times \tr(G_T)$, which is \cite[Lemma 5.12]{chenExponentialSeparationsLearning2022}. When $I=[a+b]$, $T=2$, $m=1$, $I_{1,1}=[a]$, $I_{2,1}=[a+1,\cdots, a+b]$, this lemma reduces to $\tr((G_1\otimes G_2) S_{a+b})\ge \tr(G_1 S_a)\tr(G_2 S_b)$, which is \cite[Lemma 16]{chenOptimalTradeoffsEstimating2024a}. Both lemmas are used to lower bound one-sided likelihood ratios in state discrimination problems involving Haar random states. \Cref{lem: permutation inequality} goes beyond these two special cases and can be useful in more general scenarios.
\begin{proof}
    We prove the inequality by induction on $T$. When $T=1$, the inequality is a trivial equation. We now focus on the case $T=2$. The inequality becomes
    \begin{equation}
        \tr((G_1\otimes G_2)\cdot (\otimes_{r=1}^m S^{I_{1r}\cup I_{2r}}))\ge \tr(G_1\cdot \otimes_{r=1}^m S^{I_{1r}}) \tr(G_2\cdot \otimes_{r=1}^m S^{I_{2r}}) \label{equ: permutation inequality 1}.
    \end{equation}

    For $i_r\leq \min\{\abs{I_{1r}},\abs{I_{2r}}\}$, define $S^{I_{1r}\cup I_{2r}}_{i_r}$ as the set of permutations that cross $I_{1r}$ and $I_{2r}$ $i_r$ times (in other words, the permutation maps $i_r$ qudits from $I_{1r}$ to $I_{2r}$). Then $S^{I_{1r}\cup I_{2r}} = \sum_{i_r}S^{I_{1r}\cup I_{2r}}_{i_r}$. The LHS of \eqref{equ: permutation inequality 1} is the summation of $\tr((G_1\otimes G_2)\cdot (\otimes_{r=1}^m S^{I_{1r}\cup I_{2r}}_{i_r}))$ over all possible $\{i_r\}_{r=1}^m$, while the RHS of \eqref{equ: permutation inequality 1} is $\tr((G_1\otimes G_2)\cdot (\otimes_{r=1}^m S^{I_{1r}\cup I_{2r}}_0))$. Therefore, we only need to prove that $\tr((G_1\otimes G_2)\cdot (\otimes_{r=1}^m S^{I_{1r}\cup I_{2r}}_{i_r}))\ge 0$ for all $\{i_r\}$. 

    The key observation is that every permutation in $S^{I_{1r}\cup I_{2r}}_{i_r}$ can be decomposed to three steps: we first apply some permutation on $I_{1r}$ and some permutation on $I_{2r}$. Then we swap $i_r$ qudits in $I_{1r}$ with $i_r$ qudits in $I_{2r}$. Finally apply some permutation on $I_{1r}$ and some permutation on $I_{2r}$. Therefore, up to a positive factor due to repeated counting, $S_{i_r}^{I_{1r}\cup I_{2r}}\propto (S^{I_{1r}}\otimes S^{I_{2r}})(\SWAP_r\otimes I_r)(S^{I_{1r}}\otimes S^{I_{2r}})$, where $\SWAP_r$ is a SWAP operator between $i_r$ qudits in $I_{1r}$ and $i_r$ qudits in $I_{2r}$ (the choice can be arbitrary), and $I_r$ is the identity on other qudits. Let $S_1=\otimes_{r=1}^m S^{I_{1r}}$, $S_2=\otimes_{r=1}^m S^{I_{2r}}$, and $F$ be the product of the $m$ SWAP operators (acting on $\sum_{r=1}^m i_r$ qudits in total). Let $R_1,R_2$ be the set of qudits in $P_1, P_2$ that are not swapped by $F$, respectively. We have
    \begin{align*}
        \tr((G_1\otimes G_2)\cdot (\otimes_{r=1}^m S^{I_{1r}\cup I_{2r}}_{i_r}))&\propto \tr((S_1G_1S_1^\dagger\otimes S_2G_2S_2^\dagger)\cdot (F\otimes \Id^{R_1}\otimes \Id^{R_2}))\\
        &= \tr((\tr_{R_1}(S_1G_1S_1^\dagger)\otimes \tr_{R_2}(S_2G_2S_2^\dagger))\cdot F)\\
        &= \tr(\tr_{R_1}(S_1G_1S_1^\dagger)\cdot \tr_{R_2}(S_2G_2S_2^\dagger))\ge 0.
    \end{align*}
    Therefore, the lemma holds for $T=2$. We now prove the lemma from the $T-1$ case to the $T$ case. Write $G'_1 = G_1\otimes\cdots G_{T-1}$, $G'_2 = G_{T}$, $I_{1r}'=I_{1r}\cup\cdots\cup I_{T-1, r}, I'_{2r}=I_{T, r}$. Using the results for $2$ and $T-1$,
    \begin{align*}
        \tr(G\otimes_{r=1}^m S^{\cup_{t=1}^T I_{tr}})&=\tr((G_1'\otimes G_2')\cdot (\otimes_{r=1}^m S^{I_{1r}'\cup I_{2r}'}))\\
        &\ge \tr(G_1'\cdot \otimes_{r=1}^m S^{I_{1r}'}) \tr(G_2'\cdot \otimes_{r=1}^m S^{I_{2r}'})\\
        &\ge \left(\prod_{t=1}^{T-1}\tr(G_t\otimes_{r=1}^m S^{I_{tr}})\right) \tr(G_T\otimes_{r=1}^m S^{I_{Tr}})\\
        &= \prod_{t=1}^{T}\tr(G_t\otimes_{r=1}^m S^{I_{tr}}).\qedhere
    \end{align*}
\end{proof}

\section{Lower bounds}
In this section, we leverage Haar-assembled ensembles to construct hard instances for estimating nonlinear properties and testing spectrum. In this section, all qudits have the same local dimension $d$ equal to the dimension of the unknown state $\rho$. 

\subsection{Constructing hard instances}
For two probability distributions $\{p_r\}_{r=0}^m, \{q_r\}_{r=0}^m$, define two single-qudit ensembles 
\begin{equation}
    \cE_1\coloneqq\left\{p_0\rho_0+\sum_{r=1}^m p_r\psi_r\right\}_{\psi_1, \cdots, \psi_m\leftarrow \mu_\Haar(d)}, \quad \cE_2\coloneqq\left\{q_0\rho_0+\sum_{r=1}^m q_r\psi_r\right\}_{\psi_1, \cdots, \psi_m\leftarrow \mu_\Haar(d)}, \label{equ: hard instance}
\end{equation}
where $\rho_0$ is the maximally mixed state, and $\psi_r$s are Haar random states. Compared to \eqref{equ: overview 1}, here we include an extra term of the maximally mixed state. This is used to tune the $\epsilon$-dependence.

Write $\bp=(p_1,\cdots, p_m), \bq=(q_1,\cdots, q_m)$ (note that $p_0$ and $q_0$ are not included). We will choose $p_0=q_0$ and choose $\bp, \bq$ such that the moments of $\bp$ and $\bq$ are consistent up to degree $k$, but differ at degree $k+1$. In this setting, we aim to prove that,
\begin{itemize}
    \item $\E_{\rho\leftarrow \cE_1}[\rho^{\otimes k}] = \E_{\rho\leftarrow \cE_2}[\rho^{\otimes k}]$. Furthermore, $\cE_1$ and $\cE_2$ are hard to distinguish using $k$-replica protocols.
    \item $\E_{\rho\leftarrow \cE_1}[\rho^{\otimes k+1}] \neq \E_{\rho\leftarrow \cE_2}[\rho^{\otimes k+1}]$. Therefore, we can distinguish between $\cE_1$ and $\cE_2$ if we can learn some $(k+1)$-degree properties of $\rho$ efficiently.
\end{itemize}
We formalize and prove the two points in \Cref{lem: hard instance moments} and \Cref{lem: hard instance}.

To express $\E_{\rho\leftarrow \cE_1}[\rho^{\otimes k}]$, we introduce monomial symmetric polynomials. 
\begin{definition}[Monomial symmetric polynomials]
    For $t\in \N^+$, we use $\lambda \vdash t$ to denote that $\lambda$ is a partition of $t$, that is, a multiset of positive integers $\lambda=\{\lambda_1, \cdots, \lambda_{l(\lambda)}\}$  such that $\lambda_1+\cdots +\lambda_{l(\lambda)}=t$. Denote the monomial symmetric polynomial
    \begin{equation}
        m_\lambda(\bp)\coloneqq \sum_{i_1,\cdots, i_{l(\lambda)}\in [m], \text{distinct}}p_{i_1}^{\lambda_1}\cdots p_{i_{l(\lambda)}}^{\lambda_{l(\lambda)}}. \label{equ: monomial symmetric polynomial}
    \end{equation}
\end{definition}
Expanding $\E_{\rho\leftarrow \cE_1}[\rho^{\otimes k}]$, each term is of the form 
\begin{equation}
    \rho_0^{\otimes A_0}\otimes \E_{\psi_1}[\psi_1^{\otimes A_1}]\otimes\cdots \otimes \E_{\psi_m}[\psi_m^{\otimes A_m}]=\rho_0^{\otimes A_0}\otimes \tilde{S}^{A_1}\otimes \cdots \otimes \tilde{S}^{A_m}, \label{equ: expanding hard instance 1}
\end{equation}
where we write $\tilde{S}^{A}\coloneqq \frac{1}{d^{\uparrow \abs{A}}}S^A$. Notice that this formula is not changed if we permute $\psi_1, \cdots, \psi_m$. After collecting like terms, the coefficient of \eqref{equ: expanding hard instance 1} is $p_0^{\abs{A_0}}$ multiplying a symmetric polynomial in $\bp$. To explicitly write down the polynomial, we select nonempty subsets in $A_1,\cdots, A_m$. The cardinalities of them give a partition $\lambda$ of $k-\abs{A_0}$. The coefficient of \eqref{equ: expanding hard instance 1} is $p_0^{\abs{A_0}}m_{\lambda}(\bp)$. Therefore,
\begin{equation}
    \E_{\rho\leftarrow \cE_1}[\rho^{\otimes k}] = \sum_{a=0}^k\sum_{\lambda\vdash k-a}p_0^a m_{\lambda}(\bp)\sum_{\substack{A, B_1,\cdots, B_{l(\lambda)} \text{ partition } [k]\\ \abs{A}=a, \abs{B_i}=\lambda_i}}\rho_0^{\otimes A}\otimes \tilde{S}^{B_1}\otimes \cdots \otimes \tilde{S}^{B_{l(\lambda)}}. \label{equ: expanding hard instance 2}
\end{equation}
Since $m_{\lambda}(\bp)$ is a symmetric polynomial, it can be written as a polynomial in elementary symmetric polynomials, which in turn can be written as a polynomial in power sums by Newton's identities. The full formula is complicated. Here we provide a self-contained way to calculate the leading coefficient. The proof is delayed to \Cref{sec: proof of lem: monomial to power sum}.

\begin{lemma}\label{lem: monomial to power sum}
    Let $t\in \mathbb{N}$ and let $\lambda=\{\lambda_1,\cdots,\lambda_l\}$ be a partition of $t$. 
    Denote by $m_{\lambda}(\bp)$ the monomial symmetric polynomial defined in \eqref{equ: monomial symmetric polynomial}, and let $s_{u}(\bp)\coloneqq \sum_{i=1}^m p_i^u$ be the degree-$u$ power sum. 
    Then there exists a polynomial $g_{\lambda}$ such that
    \begin{equation}
        m_{\lambda}(\bp) = (-1)^{\,l-1}(l-1)!\, s_t(\bp) \;+\; g_{\lambda}(s_1(\bp), s_2(\bp),\dots, s_{t-1}(\bp)).
    \end{equation}
    Equivalently, in the expansion of $m_{\lambda}(\bp)$ in the power-sum basis, the term $s_t(\bp)$ occurs with coefficient $(-1)^{\,l-1}(l-1)!$, and all remaining terms involve only $s_1,\dots,s_{t-1}$.
\end{lemma}

With \eqref{equ: expanding hard instance 2} and \Cref{lem: monomial to power sum}, we are ready to compare the moments of $\cE_1$ and $\cE_2$ in \Cref{lem: hard instance moments}.

\begin{lemma}\label{lem: hard instance moments}
    Let $\cE_1, \cE_2$ be two single-qudit ensembles defined in \eqref{equ: hard instance}. Suppose $p_0=q_0$ and $\sum_{r=1}^m p_r^i=\sum_{r=1}^m q_r^i$ for $i\in [k]$. Then $\E_{\rho\leftarrow \cE_1}[\rho^{\otimes k}] = \E_{\rho\leftarrow \cE_2}[\rho^{\otimes k}]$, 
    $$\E_{\rho\leftarrow \cE_1}[\rho^{\otimes k+1}] - \E_{\rho\leftarrow \cE_2}[\rho^{\otimes k+1}]=(\sum_{r=1}^m p_r^{k+1}-\sum_{r=1}^m q_r^{k+1})\Delta_{k+1},$$
    where
    \begin{equation}
        \Delta_{k+1} \coloneqq \sum_{l=1}^{k+1}(-1)^{l-1}(l-1)!\sum_{\substack{B_1,\cdots, B_{l} \text{ partition } [k+1]\\ \text{all nonempty}}}\tilde{S}^{B_1}\otimes \cdots \otimes \tilde{S}^{B_{l}}. \label{equ: Delta}
    \end{equation}
\end{lemma}
\begin{proof}
    Since $s_{i}(\bp)=s_{i}(\bq)$ for $i\in [k]$, \Cref{lem: monomial to power sum} implies that 
    \begin{equation}
        m_\lambda(\bp)-m_\lambda(\bq) = \begin{cases}
            0,&\lambda\vdash k', k'\leq k\\
            (-1)^{l(\lambda)-1}(l(\lambda)-1)!(s_{k+1}(\bp)-s_{k+1}(\bq)),&\lambda\vdash k+1
        \end{cases}.
    \end{equation}
    By \eqref{equ: expanding hard instance 2}, $\E_{\rho\leftarrow \cE_1}[\rho^{\otimes k}] = \E_{\rho\leftarrow \cE_2}[\rho^{\otimes k}]$ and 
    \begin{align}
        &\E_{\rho\leftarrow \cE_1}[\rho^{\otimes k+1}] - \E_{\rho\leftarrow \cE_2}[\rho^{\otimes k+1}] \nonumber\\
        =& \sum_{\lambda\vdash k+1}(-1)^{l(\lambda)-1}(l(\lambda)-1)!(s_{k+1}(\bp)-s_{k+1}(\bq)) \sum_{\substack{B_1,\cdots, B_{l(\lambda)} \text{ partition } [k]\\ \abs{B_i}
        =\lambda_i}}\tilde{S}^{B_1}\otimes \cdots \otimes \tilde{S}^{B_{l(\lambda)}}\nonumber\\
        =&(s_{k+1}(\bp)-s_{k+1}(\bq))\sum_{l=1}^{k+1}(-1)^{l-1}(l-1)! \sum_{\substack{B_1,\cdots, B_{l} \text{ partition } [k]\\ \text{all nonempty}}}\tilde{S}^{B_1}\otimes \cdots \otimes \tilde{S}^{B_{l}}.\qedhere
    \end{align}
\end{proof}

Now we prove that $\cE_1$ and $\cE_2$ are hard to distinguish using $k$-replica protocols. 
\begin{lemma}\label{lem: hard instance}
    Let $\{p_r\}_{r=0}^m, \{q_r\}_{r=0}^m$ be two probability distributions such that $p_0 = q_0$ and $\sum_{r=1}^m p_r^i=\sum_{r=1}^m q_r^i$ for $i\in [k]$. Define two single-qudit ensembles $\cE_1, \cE_2$ by 
    \begin{equation}
        \cE_1\coloneqq\{p_0\rho_0+\sum_{r=1}^m p_i\psi_i\}_{\psi_1, \cdots, \psi_m\leftarrow \mu_\Haar(d)}, \quad \cE_2\coloneqq\{q_0\rho_0+\sum_{r=1}^m q_i\psi_i\}_{\psi_1, \cdots, \psi_m\leftarrow \mu_\Haar(d)},
    \end{equation}
    where $\rho_0$ is the maximally mixed state, $\psi_1,\cdots, \psi_m$ are Haar random. Then the sample complexity of distinguishing $\cE_1$ and $\cE_2$ with a constant success probability using $k$-replica protocol is at least $\Omega(\frac{\sqrt{d}}{1-p_0})$.
\end{lemma}
\begin{proof}
    By \Cref{lem: Le Cam}, the success probability of distinguishing $\cE_1$ and $\cE_2$ using $k$-replica protocol is at most $$\max_{M\in\cM_{d^k,T}}d_M(\E_{\rho\leftarrow \cE_1}[\rho^{\otimes kT}],\E_{\rho\leftarrow \cE_2}[\rho^{\otimes kT}]).$$
    Notice that 
    \begin{equation}
        \{\rho^{\otimes k}\}_{\rho\leftarrow \cE_1} = \left\{\sum_{A_0,A_1,\cdots, A_m\text{ partition }[k]}\parens*{\prod_{r=0}^{m}p_r^{\abs{A_r}}}\rho_0^{\otimes A_0}\otimes \psi_1^{\otimes A_1}\otimes \cdots \otimes \psi_m^{\otimes A_m}\right\}
    \end{equation}
    is a $k$-qudit Haar-assembled ensemble (where the rotation $U_j$ is used to allocate positions of $\psi_r$'s). We now apply \Cref{thm: Haar-assembled ensemble vs mean} to upper bound $d_M(\E_{\rho\leftarrow\cE_1}[\rho^{\otimes kT}], \E_{\rho\leftarrow\cE_1}[\rho^{\otimes k}]^{\otimes T})$.
    Recall that in \eqref{equ: haar assembled ensemble vs mean}, the ``$a_r$'' is the average occurrence of $\psi_r$ and the ``$a'_r$'' is the average squared occurrence of $\psi_r$. Here, with \Cref{cor: expectation under multinomial distribution quadratic}, they become
    \begin{equation}
        a_r = \sum_{b_0+b_1+\cdots + b_m=k}\binom{k}{b_0,\cdots, b_{m}}(\prod_{r=0}^{m}p_r^{b_r})b_r = kp_r,
    \end{equation}
    \begin{equation}
        a'_r = \sum_{b_0+b_1+\cdots + b_m=k}\binom{k}{b_0,\cdots, b_{m}}(\prod_{r=0}^{m}p_r^{b_r})b_r^2 = k(k-1)p_r^2+kp_r.
    \end{equation}
    Plugging them into \eqref{equ: haar assembled ensemble vs mean}, we obtain
    \begin{align}
        &\max_{M\in\cM_{d^k,T}}d_M(\E_{\rho\leftarrow\cE_1}[\rho^{\otimes kT}], \E_{\rho\leftarrow\cE_1}[\rho^{\otimes k}]^{\otimes T})\nonumber\\
        \leq & \frac{T(T-1)}{d}\sum_{r=1}^m(kp_r)^2 + \frac{T}{d}\sum_{r=1}^m(k(k-1)p_r^2+kp_r)\nonumber\\
        = & \frac{k^2T^2-kT}{d}\sum_{r=1}^mp_r^2 + \frac{Tk}{d}\sum_{r=1}^m p_r \nonumber\\
        \leq & \frac{1}{d}\parens*{(kT\sum_{r=1}^mp_r)^2 + (kT\sum_{r=1}^mp_r)}.
    \end{align}
    Similarly, 
    \begin{equation}
        \max_{M\in\cM_{d^k,T}}d_M(\E_{\rho\leftarrow\cE_2}[\rho^{\otimes kT}], \E_{\rho\leftarrow\cE_2}[\rho^{\otimes k}]^{\otimes T})\leq \frac{1}{d}\parens*{(kT\sum_{r=1}^mp_r)^2 + (kT\sum_{r=1}^mp_r)}.
    \end{equation}
    By \Cref{lem: hard instance moments}, $\E_{\rho\leftarrow \cE_1}[\rho^{\otimes k}]=\E_{\rho\leftarrow \cE_2}[\rho^{\otimes k}]$. By triangle inequality, we obtain
    \begin{equation}
        \max_{M\in\cM_{d^k,T}}d_M(\E_{\rho\leftarrow\cE_1}[\rho^{\otimes kT}], \E_{\rho\leftarrow\cE_2}[\rho^{\otimes kT}])\leq \frac{2}{d}\parens*{(kT\sum_{r=1}^mp_r)^2 + (kT\sum_{r=1}^mp_r)}. \label{equ: two ensembles distinguish}
    \end{equation}
    According to \Cref{lem: reduction to TV}, the sample complexity of distinguishing $\cE_1$ and $\cE_2$ with a constant success probability using $k$-replica protocols is at least $\Omega(\sqrt{d}/\sum_{r=1}^m p_r)= \Omega(\sqrt{d}/(1-p_0))$.
\end{proof}

\subsection{Applications in estimating $(k+1)$-body observables}\label{sec: k+1 body observables}
We now apply \Cref{lem: hard instance moments} and \Cref{lem: hard instance} to prove lower bounds for estimating $(k+1)$-body observables.

\begin{theorem}\label{thm: lower bound estimating observables}
    Let $k, d\in \N^+$, $O$ be a $(k+1)$-qudit observable with the operator norm $\norm{O}_{\infty}\leq 1$ and the nontrivial $(k+1)$-body weight $w(O)\coloneqq \abs{\sum_{\pi\in S_{k+1}^{\text{circ}}}c_\pi(O)}\ge 2(k+1)^{2k}/d$. Let $\epsilon$ be the error parameter such that $30(k+1)^3d^{-1/2}\leq \epsilon \leq (2(k+1))^{-k}w(O)/5$. 
    Given access to a $d$-dimensional mixed quantum state $\rho\in \mathbb{C}^{d\times d}$, the sample complexity of estimating $\tr(\rho^{\otimes k+1}O)$ within additive error $\epsilon$ with probability $0.9$ is at least $\Omega(\frac{\sqrt{d}}{k\epsilon^{1/(1+k)}})$ for any $k$-replica protocol.
\end{theorem}

We need the following two lemmas. The proofs of them are left in \Cref{sec: auxiliary}. 
\begin{lemma}\label{lem: equal moment sequences}
    For $k\in \N^*$, there exists two probability distributions $\{p_i\}_{r=1}^{k+1}$ and $\{q_i\}_{r=1}^{k+1}$ such that $\sum_{r=1}^{k+1} p_r^i = \sum_{r=1}^{k+1}q_r^i$ for all $0\leq i\leq k$, and $\sum_{r=1}^{k+1} p_i^{k+1}-\sum_{r=1}^{k+1}q_i^{k+1}= \frac{2}{2^k(k+1)^k}$.
\end{lemma}

\begin{lemma}\label{lem: O component}
    Let $O$ be a $(k+1)$-qudit observable and $\Delta_k$ be the operator defined in \eqref{equ: Delta}. Then we have 
    \begin{equation}
        \abs{\sum_{\pi\in S_{k+1}^{\text{circ}}}c_\pi(O)-\tr(O\Delta_{k+1})}\leq \frac{(k+1)^k (k+1)!}{d}\leq \frac{(k+1)^{2k}}{d}.
    \end{equation}
\end{lemma}

\begin{proof}[Proof of \Cref{thm: lower bound estimating observables}]
    Write $w\coloneqq w(O)$ and $\delta_k\coloneqq (2(k+1))^{-k}$. Let $\{p_r'\}_{r=1}^{k+1}$ and $\{q_r'\}_{r=1}^{k+1}$ be the sequence given by \Cref{lem: equal moment sequences}. Let $p_0=q_0=1-(5\epsilon/(w\delta_k))^{1/(k+1)}$, $p_r=(5\epsilon/(w\delta_k))^{1/(k+1)}p_r'$, and $q_r=(5\epsilon/(w\delta_k))^{1/(k+1)}q_r'$. Then $\sum_{r=1}^{k+1}p_r^{k+1}-\sum_{r=1}^{k+1}q_r^{k+1}=10\epsilon/w$. 

    Consider the two single-qudit ensembles $\cE_1,\cE_2$ in \Cref{lem: hard instance} with the $\{p_r\}_{r=0}^{k+1}, \{q_r\}_{r=0}^{k+1}$ defined above. By \Cref{lem: hard instance}, to distinguish between $\cE_1$ and $\cE_2$ with probability at least $0.6$, the sample complexity is at least $\Omega(\frac{\sqrt{d}}{1-p_0})=\Omega(\frac{\sqrt{d}}{k\epsilon^{1/(1+k)}})$.

    Now assume we can estimate $\tr(\rho^{\otimes k+1}O)$ within additive error $\epsilon$ with probability $0.9$. Denote the estimator by $\hat{E}$ and let $E_1\coloneqq \E_{\rho'\leftarrow \cE_1}[\tr(\rho'^{\otimes k+1}O)]$, $E_2\coloneqq \E_{\rho'\leftarrow \cE_2}[\tr(\rho'^{\otimes k+1}O)]$. We design a simple algorithm that distinguishes $\cE_1$ and $\cE_2$: Check whether the estimator of $\tr(\rho^{\otimes k+1}O)$ is closer to $E_1$ or $E_2$, and output the corresponding ensemble. We now show that this method succeeds in distinguishing $\cE_1$ and $\cE_2$ with probability at least 0.8.

    Without loss of generality, suppose the ground truth is $\cE_1$. By definition, with probability at least 0.9, 
    \begin{equation}
        \abs{\hat{E}-\tr(\rho^{\otimes k+1}O)}\leq \epsilon.
    \end{equation}
    Since $\{\rho^{\otimes k+1}\}_{\rho\leftarrow \cE_1}$ is a $(k+1)$-qudit Haar-assembled ensemble, by the concentration of Haar-assembled ensemble (\Cref{thm: concentration of Haar-assembled ensemble}, where $m=k+1$ and the average occurrence $a_r$ of each $\psi_r$ is at most $k+1$), 
    \begin{equation*}
        \Pr_{\rho\leftarrow \cE_1}\bracks*{{\abs{\tr(\rho^{\otimes k+1}O)- E_1}}>\epsilon}\leq 2(k+1)\exp(-\frac{d\epsilon^2}{18\pi^3(k+1)^4\norm{O}_{\infty}})\leq 2(k+1)\exp(-(k+1)^2)< 0.1.
    \end{equation*}
    Therefore, with probability at least $0.9-0.1=0.8$, we have $\abs{\hat{E}-E_1}\leq 2\epsilon$. 

    On the other hand,
    \begin{align}
        \abs{E_1-E_2}&=\abs{\sum_{r=1}^{k+1}p_r^{k+1}-\sum_{r=1}^{k+1}q_r^{k+1}}\cdot \abs{\tr(O\Delta_{k+1})}\tag*{(\Cref{lem: hard instance moments})}\\
        &\ge \frac{10\epsilon}{w}\cdot (w-\frac{(k+1)^{2k}}{d})\tag*{(\Cref{lem: O component})}\\
        &\ge \frac{10\epsilon}{w}\cdot (w-w/2) \tag{$w\ge 2(k+1)^{2k}/d$}\\
        &= 5\epsilon.
    \end{align}

    Therefore, with probability at least 0.8, $\hat{E}$ is closer to $E_1$ and our algorithm gives the correct answer. So an algorithm that estimates $\tr(\rho^{\otimes k+1}O)$ within additive error $\epsilon$ with probability 0.9 will solve the distinguishing task with probability at least $0.8$. But we already know that the distinguishing task requires $\Omega(\frac{\sqrt{d}}{k\epsilon^{1/(1+k)}})$ samples for $k$-replica protocols. Hence estimating $\tr(\rho^{\otimes k+1}O)$ requires $\Omega(\frac{\sqrt{d}}{k\epsilon^{1/(1+k)}})$ samples for $k$-replica protocols.
\end{proof}

\Cref{thm: lower bound estimating observables} captures the most general degree $(k+1)$ properties, $\tr(\rho^{\otimes k+1}O)$ for a $d^{k+1}$-dimensional observable $O$, but with the restriction that $O$ has a large nontrivial $(k+1)$-body weight. 

Now we discuss its implication in a special case, $\tr(\rho^{k+1}O)$, for a $d$-dimensional observable $O$. Let $O'=\frac12[(O\otimes \Id_{d^k})\pi+\pi^\dagger (O\otimes \Id_{d^k})]$ for any $\pi\in S_{k+1}^{\text{circ}}$. It is easy to see that $\norm{O'}_{\infty}\leq \norm{O}_{\infty}$, $\tr(\rho^{\otimes k+1}O')=\tr(\rho^{k+1}O)$ and $w(O')=\frac{\abs{\tr(O)}}{d}$ because
\begin{align}
    &\E_{U\leftarrow \mu_\Haar}[U^{\otimes k+1}O' (U^\dagger)^{\otimes k+1}]\nonumber\\
    =&\frac12[(\E[UOU^\dagger]\otimes \Id_{d^k})\pi + \pi^\dagger(\E[UOU^\dagger]\otimes \Id_{d^k})] \tag*{($U^{\otimes k+1}$ commutes with $\pi$)}\\
    =& \frac{\tr(O)}{2d}(\pi+\pi^\dagger) \tag*{($\E_{U}[UOU^\dagger]=\tr(O)\Id/d$)}.
\end{align}
Apply \Cref{thm: lower bound estimating observables} to this $O'$, we immediately obtain the following result.
\begin{lemma}\label{lem: lower bound estimating observables 2}
    Let $k, d\in \N^+$, $O$ be a $d$-dimensional observable with the operator norm $\norm{O}_{\infty}\leq 1$ and $\abs{\tr(O)}\ge 2(k+1)^{2k}$. Let $\epsilon$ be the error parameter such that $30(k+1)^3d^{-1/2}\leq \epsilon \leq (2(k+1))^{-k}\abs{\tr(O)}/(5d)$. 
    Given access to a $d$-dimensional mixed quantum state $\rho\in \mathbb{C}^{d\times d}$, the sample complexity of estimating $\tr(\rho^{k+1}O)$ within additive error $\epsilon$ with probability $0.9$ is at least $\Omega(\frac{\sqrt{d}}{k\epsilon^{1/(1+k)}})$ for any $k$-replica protocol.
\end{lemma}
However, this result is not satisfactory because it require $\abs{\tr(O)}=\Omega(d)$, so it is not applicable to traceless observables including Pauli observables. The reason is that the positive and negative parts of $O$ cancel out when we calculate $\E_{\rho\leftarrow \cE_1}[\tr(\rho^{k+1}O)]-\E_{\rho\leftarrow \cE_2}[\tr(\rho^{k+1}O)]$. A simple way to circumvent this issue is to construct the ensembles in a subspace spanned by eigenvectors with large (or small) eigenvalues. This idea leads to our main result. 

\begingroup
\def\thetheorem{\ref{thm: main theorem}}
\begin{theorem}\label{thm: lower bound estimating observables 3}
    Let $k, d\in \N^+$, $O$ be a $d$-dimensional observable with the operator norm $\norm{O}_{\infty}\leq 1$ and $\norm{O}_1\ge 2(k+1)^{2k}$. Let $\epsilon$ be the error parameter such that $50(k+1)^3d^{-1/2}\leq \epsilon \leq (2(k+1))^{-k}\norm{O}_1/(10d)$. 
    Given access to a $d$-dimensional mixed quantum state $\rho\in \mathbb{C}^{d\times d}$, the sample complexity of estimating $\tr(\rho^{k+1}O)$ within additive error $\epsilon$ with probability $0.9$ is at least $\Omega(\frac{\sqrt{d}}{k\epsilon^{1/(1+k)}})$ for any $k$-replica protocol.
\end{theorem}
\addtocounter{theorem}{-1}
\endgroup

\begin{proof}
    The first step is to find $d/2$ eigenvalues of $O$ such that the absolute value of their sum is at least $\norm{O}_1/4$.
    Without loss of generality, we assume that the number of nonnegative eigenvalues of $O$ is at least $d/2$. Sort the eigenvalues of $O$ as $\lambda_1\geq \lambda_2\geq \cdots \geq \lambda_s \ge 0 > \lambda_{s+1} \geq \cdots \geq \lambda_d$, where $s\ge d/2$. If $\sum_{i=1}^{d/2}\lambda_i \geq \norm{O}_1/4$, we are done. Otherwise if $\sum_{i=1}^{d/2}\lambda_i < \norm{O}_1/4$, then $$\abs{\sum_{i=d/2+1}^{d}\lambda_i}\ge \sum_{i=s+1}^{d}\abs{\lambda_i}-\sum_{i=d/2}^s \lambda_i=\norm{O}_1 - \sum_{i=1}^{d/2}\lambda_i - 2\sum_{i=d/2}^s \lambda_i\ge \norm{O}_1-3\sum_{i=1}^{d/2}\lambda_i>\frac{\norm{O}_1}{4}.$$
    Therefore, we can always find $d/2$ eigenvalues of $O$ such that the absolute value of their sum is at least $\norm{O}_1/4$. Let $O_0$ be the diagonal matrix whose diagonal entries are these $d/2$ eigenvalues and $O_1$ be the diagonal matrix whose diagonal entries are the remaining $d/2$ eigenvalues. Then $UOU^\dagger=O_0\oplus O_1$ for some unitary $U$.

    By definition, $\abs{\tr(O_0)}\ge \norm{O}_1/4$, so $\norm{O}_1/(10d)\leq \abs{\tr(O_0)}/(5(d/2))$. Applying \Cref{lem: lower bound estimating observables 2} to $O_0$ (where the local dimensional is $d/2$), we know that for $\abs{\tr(O_0)}\ge 2(k+1)^{2k}$ and $30k^3(d/2)^{-1/2}\leq \epsilon \leq (2(k+1))^{-k}\norm{O}_1/(10d)$, the sample complexity of estimating $\tr(\rho_0^{k+1}O_0)$ within additive error $\epsilon$ with probability $0.9$ is at least $\Omega(\frac{\sqrt{d/2}}{k\epsilon^{1/(1+k)}})=\Omega(\frac{\sqrt{d}}{k\epsilon^{1/(1+k)}})$ for any $k$-replica protocol, where $\rho_0$ is a $d/2$-dimensional mixed quantum state.

    Now assume that there is a $k$-replica protocol $\cA$ that estimates $\tr(\rho^{k+1}O)$. We construct a $k$-replica protocol $\cA_0$ that estimates $\tr(\rho_0^{k+1}O_0)$ with the same precision, success probability, and sample complexity. 
    $\cA_0$ works as follows: given access to $(d/2)$-dimensional state $\rho_0$, it construct a $d$-dimensional state $\rho=U(\rho_0\oplus \Zero_{d/2})U^\dagger$, where $\Zero_{d/2}$ is the $d/2$-dimensional zero matrix. Then it runs $\cA$ on $\rho$. Since
    \begin{equation}
        \tr(\rho^{k+1}O) = \tr(U^\dagger (\rho_0\oplus \Zero_{d/2})^{k+1} U O)=\tr((\rho_0^{k+1}\oplus \Zero_{d/2}) (O_0\oplus O_1))=\tr(\rho_0^{k+1}O_0),
    \end{equation}
    the output of $\cA$ is also an estimate of $\tr(\rho_0^{k+1}O_0)$ with the same precision and success probability. According to the hardness of estimating $\tr(\rho_0^{k+1}O_0)$, the sample complexity of estimating $\tr(\rho^{k+1}O)$ is at least $\Omega(\frac{\sqrt{d}}{k\epsilon^{1/(1+k)}})$ for any $k$-replica protocol.
\end{proof}

\subsection{Applications in spectrum testing}\label{sec: spectrum testing}
As another illustrative application of our framework, we establish lower bounds for spectrum testing and rank testing under $k$-replica protocols.
A $d$-dimensional spectrum $\bp=(p_1,\dots,p_d)$ refers to a probability distribution with entries arranged in nonincreasing order (sometimes the length of the sequence is smaller than $d$, in which case we implicitly pad zeros to the end of the sequence).
The rank of a spectrum is the number of its nonzero entries, and the rank of a quantum state is defined as the rank of its spectrum.

\begin{task}[$\bp$-versus-$\bq$ spectrum testing]
    Given copy access to an unknown state $\rho$ whose spectrum is promised to be either $\bp$ or $\bq$, determine which case holds with probability at least $0.9$.
\end{task}

When $\bp$ and $\bq$ agree on all moments up to degree $k$, a natural distinguishing strategy is to estimate the $(k+1)$-th moments, which requires $(k+1)$-replica measurements. Here we show that $(k+1)$-replica measurements are indeed necessary.

\begin{theorem}\label{thm: spectrum testing}
    Let $\bp=(p_1,\cdots, p_d)$ and $\bq=(q_1,\cdots q_d)$ be two probability distributions such that $\sum_{r=1}^d p_r^i=\sum_{r=1}^d q_r^i$ for $i\in [k]$. The sample complexity of $\bp$-versus-$\bq$ spectrum testing is at least $\Omega(\frac{\sqrt{d}}{m\sqrt{\ln m}})$ for any $k$-replica protocol, where $m$ is the maximum rank of $\bp$ and $\bq$. 
\end{theorem}

It seems a direct corollary of \Cref{lem: hard instance} (where $p_0,q_0$ are set to 0). However, the problem is that the spectrum of $\sum_{r=1}^d p_i\psi_i$ is not exactly $\bp$. We deal with this rounding error in the following lemma.
\begin{lemma}\label{lem: round spectrum}
    Let $\ba=(a_1,\cdots, a_m)$ be a spectrum and $\psi_1,\cdots, \psi_m$ be $m$ Haar random states. With probability at least $0.99$, $\rho\coloneqq \sum_{r=1}^m a_r\psi_r$ is $\mathcal{O}(m\sqrt{\ln m}/\sqrt{d})$-close to a state with spectrum $\ba$ in trace distance.
\end{lemma}
We leave the proof to \Cref{sec: round spectrum}. 
\begin{proof}[Proof of \Cref{thm: spectrum testing}]
    Define a function $f_{\bp}$ ($f_{\bq}$, resp.)that maps a quantum state to the closest state with spectrum $\bp$ ($\bq$, resp.). Consider the following four ensembles
    \begin{equation*}
        \cE_1=\{\sum_{r=1}^d p_r\psi_r\}_{\psi_1,\cdots, \psi_d\leftarrow \mu_\Haar(d)},\,\cE_1'=\{f_{\bp}(\rho)\}_{\rho\leftarrow \cE_1}, \, \cE_2=\{\sum_{r=1}^d p_r\psi_r\}_{\psi_1,\cdots, \psi_d\leftarrow \mu_\Haar},\,\cE_2'=\{f_{\bq}(\rho)\}_{\rho\leftarrow \cE_2}.
    \end{equation*}
    For any $k$-replica $T$-round POVM $M$, \eqref{equ: two ensembles distinguish} has established that $$d_M(\E_{\rho\leftarrow\cE_1}[\rho^{\otimes kT}], \E_{\rho\leftarrow\cE_2}[\rho^{\otimes kT}])\leq \frac{4(kT)^2}{d}.$$
    By \Cref{lem: round spectrum},
    \begin{align}
        &d_M(\E_{\rho\leftarrow\cE_1}[\rho^{\otimes kT}], \E_{\rho\leftarrow\cE_1'}[\rho^{\otimes kT}])\nonumber\\
        \leq& \norm{\E_{\rho\leftarrow\cE_1}[\rho^{\otimes kT}]- \E_{\rho\leftarrow\cE_1'}[\rho^{\otimes kT}]}_1 \tag*{(property of trace norm)}\\
        \leq& \E_{\rho\leftarrow \cE_1}[\norm{\rho^{\otimes kT}-f_{\bp}(\rho)^{\otimes kT}}_1]\tag*{(triangle inequality)}\\
        \leq& 0.01 + 0.99kTm\sqrt{\ln m}/\sqrt{d} \tag*{(\Cref{lem: round spectrum} and \Cref{lem: norm facts}(d))}.
    \end{align}
    Similarly,
    \begin{equation}
        d_M(\E_{\rho\leftarrow\cE_2}[\rho^{\otimes kT}], \E_{\rho\leftarrow\cE_2'}[\rho^{\otimes kT}])\leq 0.01 + 0.99kTm\sqrt{\ln m}/\sqrt{d}.
    \end{equation}
    Applying triangle inequality,
    \begin{equation}
        d_M(\E_{\rho\leftarrow\cE_1'}[\rho^{\otimes kT}], \E_{\rho\leftarrow\cE_2'}[\rho^{\otimes kT}])\leq 0.02+\frac{4(kT)^2}{d}+2kTm\sqrt{\ln m}/\sqrt{d}.
    \end{equation}
    By Le Cam's two-point method (\Cref{lem: Le Cam}), this upper bounds the success probability of distinguishing $\cE_1'$ and $\cE_2'$ using $k$-replica protocols with $T$ rounds. To achieve a winning probability of $0.9$, the sample complexity is at least $kT=\Omega(\frac{\sqrt{d}}{m\sqrt{\ln m}})$. 

    On the other hand, assume we can solve the spectrum testing task, obviously we can solve the distinguishing task. Hence the sample complexity of the spectrum testing task is at least $\Omega(\frac{\sqrt{d}}{m\sqrt{\ln m}})$. 
\end{proof}

Now we prove \Cref{cor: rank testing} and \Cref{cor: noninteger power} by carefully choosing $\bp$ and $\bq$.

\begingroup
\def\thetheorem{\ref{cor: rank testing}}
\begin{corollary}
    Let $k, d\in \N^+$ and $\epsilon=2/(k+1)^3$. The sample complexity of $(k, \epsilon)$ rank testing is at least $\Omega(\sqrt{d}/(k\sqrt{\ln k}))$ for any $k$-replica protocol.
\end{corollary}
\addtocounter{theorem}{-1}
\endgroup
\begin{proof}
    Let $\bp=(p_1,\cdots, p_{k+1})$ and $\bq=(q_1,\cdots, q_{k+1})$ be two distributions obtained from \Cref{lem: equal moment sequences}. We sort them in the nonincreasing order. By definition, their moments coincide up to degree $k$. By \Cref{thm: spectrum testing}, the sample complexity to distinguish them is at least $\Omega(\sqrt{d}/(k\sqrt{\ln k}))$. 

    According to the proof of \Cref{lem: equal moment sequences}, one of $\bp$ and $\bq$ has rank at most $k$. Without loss of generality assume $\bp$ has rank at most $k$. Another fact from the proof of \Cref{lem: equal moment sequences} is that all nonzero entries of $\bq$ are at least $\epsilon$. For any state $\rho$ with spectrum $\bq$ and state $\sigma$ with rank at most $k$, writing the spectrum of $\br$ as $\br=(r_1,\cdots, r_k, r_{k+1}=0)$, we have 
    \begin{equation}
        d_{\tr}(\rho, \sigma)=\frac{1}{2}\norm{\rho-\sigma}_1\ge \frac12\sum_{i=1}^{k+1}\abs{q_i-r_i}\ge \frac12 \abs{\sum_{i=1}^{k}(q_i-r_i)}+\frac12 q_{k+1} = q_{k+1} \ge \epsilon.
    \end{equation}
    Therefore, a state $\rho$ with spectrum $\bp$ has rank at most $k$ and a state $\rho$ with spectrum $\bq$ is $\epsilon$-far away from any state with rank at most $k$. So the $\bp$-versus-$\bq$ spectrum testing can be reduced to the $(k, \epsilon)$ rank testing. Thus, we get the $\Omega(\sqrt{d}/(k\sqrt{\ln k}))$ lower bound for the $(k, \epsilon)$ rank testing problem.
\end{proof}

\begingroup
\def\thetheorem{\ref{cor: noninteger power}}
\begin{corollary}
    Let $d, m, k\in \N^+$ and $\alpha>0$ be a non-integer. Define $\epsilon(\alpha, k, m)$ as the maximum of $\frac12 (\sum_{r=1}^{m}p_r^\alpha -  \sum_{r=1}^m q_r^\alpha)$ over all pairs of distributions $\{p_r\}_{r=1}^m, \{q_r\}_{r=1}^m$ such that $\sum_{r=1}^m p_r^i=\sum_{r=1}^m q_r^i, \forall i\in [k]$. Given copy access of an unknown $d$-dimensional state $\rho$, the sample complexity of estimating $\tr(\rho^{\alpha})$ within error $\epsilon(\alpha, k, m)$ is at least $\Omega(\sqrt{d}/(m\sqrt{\ln m}))$.
\end{corollary}
\addtocounter{theorem}{-1}
\endgroup
\begin{proof}
    Let $\bp=\{p_r\}_{r=1}^m, \bq=\{q_r\}_{r=1}^m$ be the pair distributions that achieves the maximum in the definition of $\epsilon(\alpha, k, m)$. By \Cref{thm: spectrum testing}, the sample complexity of $\bp$-versus-$\bq$ spectrum testing is at least $\Omega(\sqrt{d}/(m\sqrt{\ln m}))$ for $k$-replica protocols. Since $\sum_{r=1}^m p_r^\alpha - \sum_{r=1}^m q_r^\alpha=2\epsilon(\alpha, k, m)$, estimating $\tr(\rho^\alpha)$ within error $\epsilon(\alpha, k, m)$ could solve the $\bp$-versus-$\bq$ spectrum testing task. Therefore, the sample complexity of estimating $\tr(\rho^\alpha)$ within error $\epsilon(\alpha, k, m)$ is at least $\Omega(\sqrt{d}/(m\sqrt{\ln m}))$ for $k$-replica protocols.
\end{proof}

\section{Upper bounds}\label{sec: upper bounds}
In this section, we prove the upper bound results \Cref{thm: upper bound}.

\begingroup
\def\thetheorem{\ref{thm: upper bound}}
\begin{theorem}
Let $k,d\in \N^+$, $\epsilon>0$, and $O$ be a $d$-dimensional observable with $\norm{O}_\infty\leq 1$. There exists a $\ceil{\frac{k+1}{2}}$-replica protocol that estimates $\tr(\rho^{k+1}O)$ within error $\epsilon$ with sample complexity $\tilde{\mathcal{O}}\left(\max\{\frac{k\sqrt{d}}{\epsilon}, \frac{k}{\epsilon^2}\}\right)$.
\end{theorem}
\addtocounter{theorem}{-1}
\endgroup

\begin{proof}
    Denote $k_1=\ceil{\frac{k+1}{2}}$ and $k_2=\floor{\frac{k+1}{2}}$. Then $k_1+k_2=k+1$. 
    Define two states $\rho_1 \coloneqq \frac{\rho+\rho^{k_1}}{1+\tr(\rho^{k_1})}$ and $\rho_2 \coloneqq \frac{\rho+\rho^{k_2}}{1+\tr(\rho^{k_2})}$.

    Apply the generalized swap test to $\rho^{\otimes k_1}$ and a $k_1$-qudit cyclic permutation operator $U=\pi$, and trace out the last $(k_1-1)$ qudit. With probability $\frac{1+\tr(\rho^{k_1})}{2}\ge \frac12$, we obtain the (unnormalized) state $\tr_{>1}((\Id+\pi)\rho^{\otimes k_1}(\Id+\pi^\dagger))=2\rho + 2\rho^{k_1}$, which is $\rho_1$ after normalization. Therefore, we obtain a copy of state $\rho_1$ with probability at least $1/2$. With high probability, we successfully obtain $M=\tO(\max\{\frac{\sqrt{d}}{\epsilon}, \frac{1}{\epsilon^2}\})$ copies of $\rho_1$ after $3M$ runs. Similarly, with high probability, we successfully obtain $M$ copies of $\rho_2$.

    If $k=1$, we can estimate $\tr(\rho^2O)$ using \Cref{lem: inner product estimation} directly. If $k=2$, then $\rho_1=(\rho+\rho^2)/(1+\tr(\rho^2))$. Since $\tr(\rho^3O)=(1+\tr(\rho^2))\tr(\rho_1\rho O)-\tr(\rho^2O)$, once we get estimators for $\tr(\rho^2), \tr(\rho_1\rho O)$, and $\tr(\rho^2)$ within error $\epsilon/10$, we can obtain an estimator of $\tr(\rho^3O)$ within error $\epsilon$. Here $\tr(\rho^2)$ can be estimated using the generalized swap test, $\tr(\rho_1\rho O)$ can be estimated in a distributed way using \Cref{lem: inner product estimation}, and $\tr(\rho^2O)$ can be estimated directly as it is a linear property of $\rho^{\otimes 2}$. All of these estimations can be done by $2$-replica protocols. Therefore, we obtain a $2$-replica protocol for estimating $\tr(\rho^3O)$. 

    Now we assume $k\ge 3$. The idea is similar. Notice that
    $$\tr(\rho^{k+1}O)=(1+\tr(\rho^{k_1}))(1+\tr(\rho^{k_2}))\tr(\rho_1\rho_2O)-\tr(\rho^2O)-\tr(\rho^{k_1+1}O)-\tr(\rho^{k_2+1}O).$$
    Once we get estimators for $\tr(\rho^{k_1}), \tr(\rho^{k_2}), \tr(\rho_1\rho_2O), \tr(\rho^2O), \tr(\rho^{k_1+1}O)$, and $\tr(\rho^{k_2+1}O)$ within error $\epsilon/10$, we can obtain an estimator of $\tr(\rho^{k+1}O)$ within error $\epsilon$. Here $\tr(\rho^{k_1})$, $\tr(\rho^{k_2})$, and $\tr(\rho^2O)$ can be estimated using the generalized swap test. $\tr(\rho_1\rho_2O)$ can be estimated in a distributed way using \Cref{lem: inner product estimation}. Therefore, estimating $\tr(\rho^{k+1}O)$ within error $\epsilon$ reduces to estimating $\tr(\rho^{k_1+1}O)$ and $\tr(\rho^{k_2+1}O)$ within error $\epsilon/10$, which have the same form as the original problem but with smaller degrees. If $k=3$, then $k_1=k_2=2$, and we have already shown how to estimate $\tr(\rho^3O)$ using a $2$-replica protocol. If $k\ge 4$, we repeat this reduction one more time. So the problem reduces to estimating $\tr(\rho^{k'+1}O)$ for $k'=\ceil{\frac{k_1+1}{2}}, \floor{\frac{k_1+1}{2}}, \ceil{\frac{k_2+1}{2}}$, and $\floor{\frac{k_2+1}{2}}$ within error $\epsilon/100$. Since $k\ge 4$, we can verify that $k'+1\leq k_1$. Therefore, all these $\tr(\rho^{k'+1}O)$ can be estimated using the generalized SWAP test algorithm. 

    In practice, we do not need to wait until all copies of $\rho_1, \rho_2$ are successfully generated. Instead, we measure each copy immediately after it is generated. Therefore, the overall algorithm is a $k_1$-replica protocol. The sample complexity is $\tO(\max\{\frac{k\sqrt{d}}{\epsilon}, \frac{k}{\epsilon^2}\})$ due to the cost of distributed inner product estimation. 
\end{proof}

\section{Proofs of auxiliary lemmas} \label{sec: auxiliary}
\subsection{Expectations under multinomial distributions}
In this paper, we run into the following calculation several times
\begin{equation}
    \sum_{x_1+\cdots +x_m = N}\binom{N}{x_1,\cdots, x_m} p_1^{x_1}\cdots p_m^{x_m} f(x_1,\cdots, x_m), \label{equ: expectation under multinomial distribution}
\end{equation}
where $\{p_1,\cdots, p_m\}$ is a probability distribution and $f$ is a polynomial. This can be understood as the expectation of a random variable $f(x_1, \cdots, x_m)$ when $(x_1,\cdots, x_m)$ follows a multinomial distribution. The following lemma is useful to compute such expectations.
\begin{lemma}\label{lem: expectation under multinomial distribution}
    Let $c_1,\cdots, c_m \in \N$. We have 
    \begin{equation}
        \sum_{x_1+\cdots +x_m=N}\binom{N}{x_1,\cdots, x_m}\parens*{\prod_{i=1}^m p_i^{x_i}} \prod_{i=1}^m x_i^{\downarrow c_i} = N^{\downarrow c_1+\cdots + c_m}\prod_{i=1}^m p_i^{c_i}.
    \end{equation}
\end{lemma}
\begin{proof}
    Notice that $\binom{N}{x_1,\cdots, x_m}\prod_{i=1}^m x_i^{\downarrow c_i} = N^{\downarrow c_1+\cdots + c_m}\binom{N-(c_1+\cdots + c_m)}{x_1-c_1, \cdots, x_m-c_m}$. Let $x_i' = x_i-c_i$. We have
        \begin{align*}
            &\sum_{x_1+\cdots +x_m=N}\binom{N}{x_1,x_2,\cdots, x_m}\left(\prod_{i=1}^m p_i^{x_i}\right)\prod_{i=1}^m x_i^{\downarrow c_i}\\
            =&N^{\downarrow c_1+\cdots + c_m}\sum_{x_1+\cdots +x_m=N}\binom{N-(c_1+c_2+\cdots +c_m)}{x_1-c_1,x_2-c_2,\cdots, x_m-c_m}\left(\prod_{i=1}^m p_i^{x_i}\right)\\
            =& N^{\downarrow c_1+\cdots + c_m}\prod_{i=1}^m p_i^{c_i}\sum_{x_1'+\cdots +x_m'=N-(c_1+\cdots +c_m)}\binom{N-(c_1+c_2+\cdots +c_m)}{x_1',x_2',\cdots, x_m'}\left(\prod_{i=1}^m p_i^{x_i'}\right)\\
            =& N^{\downarrow c_1+\cdots + c_m}\prod_{i=1}^m p_i^{c_i}\left(p_1+\cdots + p_m\right)^{N-(c_1+\cdots + c_m)}\\
            =& N^{\downarrow c_1+\cdots + c_m}\prod_{i=1}^m p_i^{c_i}. \qedhere
        \end{align*}
\end{proof}
For a general polynomial $f(x_1,\cdots, x_m)$, we can always write it as a linear combination of $\prod_{i=1}^m x_i^{\downarrow c_i}$s and apply \Cref{lem: expectation under multinomial distribution} to compute the expectation. For convenience, we list quadratic calculations below.
\begin{lemma}\label{cor: expectation under multinomial distribution quadratic}
    When $f(x_1, \cdots, x_m)$ is $x_i^2$, $x_ix_j$ ($i\neq j$), or $x_i$, the expectation \eqref{equ: expectation under multinomial distribution} equals to $N(N-1)p_i^2+Np_i$, $N(N-1)p_ip_j$, or $Np_i$, respectively.
\end{lemma}

\subsection{Proof of Lemma \ref{lem: monomial to power sum}}\label{sec: proof of lem: monomial to power sum}
\begingroup
\def\thetheorem{\ref{lem: monomial to power sum}}
\begin{lemma}
    Let $t\in \mathbb{N}$ and let $\lambda=\{\lambda_1,\cdots,\lambda_l\}$ be a partition of $t$. 
    Denote by $m_{\lambda}(\bp)$ the monomial symmetric polynomial defined in \eqref{equ: monomial symmetric polynomial}, and let $s_{u}(\bp)\coloneqq \sum_{i=1}^m p_i^u$ be the degree-$u$ power sum. 
    Then there exists a polynomial $g_{\lambda}$ such that
    \begin{equation}
        m_{\lambda}(\bp) = (-1)^{\,l-1}(l-1)!\, s_t(\bp) \;+\; g_{\lambda}(s_1(\bp), s_2(\bp),\dots, s_{t-1}(\bp)).
    \end{equation}
    Equivalently, in the expansion of $m_{\lambda}(\bp)$ in the power-sum basis, the term $s_t(\bp)$ occurs with coefficient $(-1)^{\,l-1}(l-1)!$, and all remaining terms involve only $s_1,\dots,s_{t-1}$.
\end{lemma}
\addtocounter{theorem}{-1}
\endgroup
\begin{proof}
    We prove it by induction on $l$. When $l=1$, $m_\lambda(\bp)=m_{\{t\}}(\bp)=s_t(\bp)$, so the statement holds with $g_\lambda=0$. 

    Now suppose $l>1$ and write $\lambda=\{\lambda_1, \cdots, \lambda_l\}$. Consider 
    \begin{align}
        m_{\{\lambda_1,\cdots, \lambda_{l-1}\}}(\bp)s_{\lambda_l}(\bp) &= \sum_{i_1,\cdots, i_{l-1}\in [m], \text{distinct}}\sum_{i_l}p_{i_1}^{\lambda_1}\cdots p_{i_{l}}^{\lambda_{l}}\nonumber\\
        &=m_{\lambda}(\bp) + \sum_{i_1,\cdots, i_{l-1}\in [m], \text{distinct}}\sum_{j=1}^{l-1}p_{i_1}^{\lambda_1}\cdots p_{i_{l-1}}^{\lambda_{l-1}}\cdot p_{i_j}^{\lambda_l}\nonumber\\
        &=m_{\lambda}(\bp) + \sum_{j=1}^{l-1}m_{\lambda^{j\leftarrow l}}(\bp),
    \end{align}
    where $\lambda^{j\leftarrow l}\coloneqq \{\lambda_1,\cdots, \lambda_{j-1}, \lambda_j+\lambda_l, \lambda_{j+1}, \cdots, \lambda_{l-1}\}$ is a partition of length $l-1$ obtained by absorbing $\lambda_l$ into $\lambda_j$. By induction hypothesis, $m_{\lambda^{j\leftarrow l}}(\bp)$ is $(-1)^{l-2}(l-2)!s_t(\bp)$ plus some polynomials in $s_1,\cdots, s_{t-1}$. The left-hand side $m_{{\lambda_1,\cdots, \lambda_{l-1}}}s_{\lambda_l}(\bp)$ is also a polynomial in $s_1,\cdots, s_{t-1}$. Therefore, $m_\lambda(\bp)$ is $-(l-1)\times(-1)^{l-2}(l-2)!s_t(\bp)=(-1)^{l-1}(l-1)!s_t(\bp)$ plus some polynomials in $s_1,\cdots, s_{t-1}$.
\end{proof}

\subsection{Proof of Lemma \ref{lem: equal moment sequences}}
\begingroup
\def\thetheorem{\ref{lem: equal moment sequences}}
\begin{lemma}
    For $k\in \N^*$, there exists two probability distributions $\{p_r\}_{r=1}^{k+1}$ and $\{q_i\}_{r=1}^{k+1}$ such that $\sum_{r=1}^{k+1} p_i^m = \sum_{r=1}^{k+1}q_i^m$ for all $0\leq m\leq k$, and $\sum_{r=1}^{k+1} p_i^{k+1}-\sum_{r=1}^{k+1}q_i^{k+1}= \frac{2}{2^k(k+1)^k}$.
\end{lemma}
\addtocounter{theorem}{-1}
\endgroup
To prove the lemma, we introduce the Chebyshev polynomial. 
Let $T_{k+1}(x)$ be the $(k+1)$-th Chebyshev polynomial of the first kind, namely the polynomial that satisfies $T_{k+1}(\cos\theta)=\cos((k+1)\theta)$. $T_{k+1}(x)$ has the following properties:
\begin{enumerate}
    \item When $x\in [-1, 1]$, $T_{k+1}(x)\in [-1, 1]$.
    \item The leading two terms of $T_{k+1}(x)$ is $2^k x^{k+1}+0x^k$.
    \item Equation $T_{k+1}(x)=1$ has $k+1$ roots in $[-1, 1]$. When $k+1$ is even, the roots are $\cos\frac{2r\pi}{k+1}$ for $r=0, 1, \cdots, (k+1)/2$, with multiplicity $1$ for $r=0,(k+1)$ and multiplicity $2$ for other $r$. When $k+1$ is odd, the roots are $\cos\frac{2r\pi}{k+1}$ for $r=0, 1, \cdots, k/2$, with multiplicity $1$ for $r=0$ and multiplicity $2$ for other $r$. 
    \item Equation $T_{k+1}(x)=-1$ has $k+1$ roots in $[-1, 1]$. When $k+1$ is even, the roots are $\cos\frac{(2r+1)\pi}{k+1}$ for $r=0, 1, \cdots, (k-1)/2$ with multiplicity 2 for all $r$. When $k+1$ is odd, the roots are $\cos\frac{(2r+1)\pi}{k+1}$ for $r=0, 1, \cdots, k/2$, with multiplicity 1 for $r=k/2$ and multiplicity $2$ for other $r$.
\end{enumerate}

\begin{proof}[Proof of \Cref{lem: equal moment sequences}]
    We rescale $T_{k+1}$ by defining $T(x)\coloneqq \delta T_{k+1}((k+1)(x-1))$, where $\delta =\frac{1}{2^k(k+1)^{k+1}}$, so that 
    \begin{enumerate}
        \item The leading two terms of $T(x)$ is $x^{k+1}-x^k$,
        \item Both $T(x)-\delta$ and $T(x)+\delta$ have $(k+1)$ nonnegative roots.
    \end{enumerate}
    Let $\{p_r\}_{r=1}^{k+1}$ and $\{q_r\}_{r=1}^{k+1}$ be roots of $T(x)-\delta$ and $T(x)+\delta$, respectively. By Vieta's formulas and the properties above, $\{p_r\}_{r=1}^{k+1},\{q_r\}_{r=1}^{k+1}$ are probability distributions such that all elementary symmetric polynomials up to degree $k$ are the same. By Newton's identities, it implies that $\sum_{r=1}^{k+1} p_r^m=\sum_{r=1}^{k+1} q_r^m$ for all $0\leq m\leq k$. The difference between $(k+1)$-th moments can be bounded by 
    \begin{equation}
        \sum_{r=1}^{k+1} p_r^{k+1}-\sum_{r=1}^{k+1}q_r^{k+1}=\sum_{r=1}^{k+1}(T(p_r)-T(q_r)) =2(k+1)\delta.
    \end{equation}
    This completes the proof the lemma. Moreover, since $-1$ is the root of either $T_{k+1}(x)=-1$ or $T_{k+1}(x)=1$, one of $\{p_r\}_{r=1}^{k+1}$ and $\{q_r\}_{r=1}^{k+1}$ contains a zero entry. And the smallest nonzero entropy is 
    \begin{equation}
        \frac{1-\cos\frac{\pi}{k+1}}{k+1}=\frac{2}{k+1}\sin^2\frac{\pi}{2(k+1)}\ge \frac{2}{(k+1)^3},
    \end{equation}
    where we use the fact $\sin x\ge 2x/\pi$ for $x\in [0,\pi/2]$.
\end{proof}

\subsection{Proof of Lemma \ref{lem: O component}}
To ease the notation, here we replace $k+1$ in the original lemma by $k$.
\begingroup
\def\thetheorem{\ref{lem: O component}}
\begin{lemma}
    Let $O$ be a $k$-qudit observable and $\Delta_k$ be the operator defined in \eqref{equ: Delta}. Then we have
    \begin{equation}
        \abs{\sum_{\pi\in S_{k}^{\text{circ}}}c_\pi(O)-\tr(O\Delta_{k})}\leq \frac{k^{k-1}k!}{d}.
    \end{equation}
\end{lemma}
\addtocounter{theorem}{-1}
\endgroup
We rewrite the definition of $\Delta_k$ here for convenience
\begin{equation*}
    \Delta_{k} \coloneqq \sum_{l=1}^{k}(-1)^{l-1}(l-1)!\sum_{\substack{B_1,\cdots, B_{l} \text{ partition } [k]\\ \text{all nonempty}}}\tilde{S}^{B_1}\otimes \cdots \otimes \tilde{S}^{B_{l}}.
\end{equation*}
We introduce Stirling numbers of the second kind, $S(k, l)$, defined as the number of ways to partition a set of $k$ elements into $l$ nonempty subsets. A useful property of $S(k, l)$ is the polynomial identity $x^k=\sum_{l=1}^kS(k, l)x^{\downarrow l}$. So $x^{k-1}=\sum_{l=1}^kS(k, l)(x-1)^{\downarrow l-1}$. Setting $x=0$, we obtain that
\begin{equation}
    \sum_{l=0}^k S(k,l)(-1)^{l-1}(l-1)! = \begin{cases}
    1, &k=1\\
    0, &k>1
    \end{cases}.
\end{equation}

\begin{proof}[Proof of \Cref{lem: O component}]
    By definition of $\Delta_{k+1}$, $(U^\dagger)^{\otimes (k+1)}\Delta_{k+1}U^{\otimes (k+1)}=\Delta_{k+1}$. So
    \begin{equation}
        \tr(O\Delta_{k+1})=\E_{U\leftarrow \mu_\Haar}[\tr(U^{\otimes (k+1)}O(U^\dagger)^{\otimes (k+1)}\Delta_{k+1})]=\sum_{\pi}c_{\pi}(O)\tr(\Delta_{k+1}\pi).\label{equ: O component 1}
    \end{equation}
    The goal is to show that $\tr(\Delta_{k}\pi)\approx 1$ if $\pi\in S_{k+1}^{\text{circ}}$ and $\tr(\Delta_{k}\pi)\approx 0$ otherwise.
    Fix a $\pi$, to calculate $\tr(\Delta_k\pi)$, we divide terms $B_1,\cdots, B_l$ in $\Delta_k$ into two cases.

    \noindent\textbf{Case 1:} There exists an $i\in [k]$ such that $i$ and $\pi(i)$ are in different parts of the partition, say, $i\in B_{r_1}$ and $\pi(i)\in B_{r_2}$ for $r_1\neq r_2$.

    Apparently, such $i$ cannot be unique. So there is another $j\in [k]$ such that $j$ and $\pi(j)$ are in different parts, say, $j\in B_{r_3}$ and $\pi(j)\in B_{r_4}$ for $r_3\neq r_4$.
    In this case 
    \begin{align}
        \abs{\tr(\tilde{S}^{B_1}\otimes \cdots \otimes \tilde{S}^{B_{l}}\cdot \pi)}&\leq \E_{\psi_1,\cdots \psi_l\leftarrow\mu_\Haar}[\abs{\tr(\psi_1^{\otimes B_1}\otimes \cdots \psi_l^{\otimes B_l}\cdot \pi)}]\leq \E[\abs{\braket{\psi_{r_1}}{\psi_{r_2}}\braket{\psi_{r_3}}{\psi_{r_4}}}]. \label{equ: O component 2}
    \end{align}
    If $\{r_1,r_2\}=\{r_3,r_4\}$, then the RHS is just $\frac{1}{d}$. If $\{r_1,r_2\}\cap \{r_3, r_4\} = \emptyset$, then by Cauchy's inequality, the RHS is 
    \begin{equation*}
        \E[\abs{\braket{\psi_{r_1}}{\psi_{r_2}}}]\cdot \E[\abs{\braket{\psi_{r_3}}{\psi_{r_4}}}]\leq \sqrt{\E[\abs{\braket{\psi_{r_1}}{\psi_{r_2}}}^2]}\cdot \sqrt{\E[\abs{\braket{\psi_{r_3}}{\psi_{r_4}}}^2]}=\frac{1}{d}.
    \end{equation*}
    If $\abs{\{r_1,r_2\}\cap \{r_3, r_4\}}=1$, say $r_1=r_3$, then by Cauchy's inequality, the RHS of \eqref{equ: O component 2} is at most
    \begin{equation*}
        \sqrt{\E_{\psi_{r_1}}\E_{\psi_{r_2},\psi_{r_4}}[\abs{\braket{\psi_{r_1}}{\psi_{r_2}}}^2\cdot \abs{\braket{\psi_{r_1}}{\psi_{r_4}}}^2]}=\sqrt{\E_{\psi_{r_1}}\frac{1}{d^2}}=\frac{1}{d}.
    \end{equation*}
    In summary, the RHS of \eqref{equ: O component 2} is always at most $1/d$. So the total contribution of this case is at most 
    \begin{equation}
        \frac{1}{d}\sum_{l=1}^k(l-1)!S(k, l)\leq \frac{1}{d}\sum_{l=1}^k (k-1)^{\downarrow l-1}S(k, l)=\frac{k^{k-1}}{d}.
    \end{equation}

    \noindent\textbf{Case 2:} $i$ and $\pi(i)$ are in the same part for every $i$. 

    The cycle decomposition of $\pi$ gives a partition of $[k]$, denoted by $A_1,\cdots, A_t$. Case 2 happens if and only if each $A_j$ is contained in one of $B_j$. The number of such partitions $B_1,\cdots, B_l$ is $S(t, l)$ (the number of ways to divide $\{A_1,\cdots, A_t\}$ into $l$ parts). For each such partition $B_1, \cdots, B_l$,
    \begin{equation}
        \tr(\tilde{S}^{B_1}\otimes \cdots \otimes \tilde{S}^{B_{l}}\cdot \pi)=\E_{\psi_1,\cdots \psi_l\leftarrow\mu_\Haar}[\tr(\psi_1^{\otimes B_1}\otimes \cdots \psi_l^{\otimes B_l}\cdot \pi)]=1.
    \end{equation}
    So the total contribution of this case is
    \begin{equation}
        \sum_{l=1}^k (-1)^{l-1}(l-1)! S(t,l) = \begin{cases}
    1, &t=1\\
    0, &t>1
    \end{cases}.
    \end{equation}
    Here $t=1$ if and only if $\pi\in S_k^{\text{circ}}$. 
    In summary, we have shown that, up to error $k^{k-1}/\sqrt{d}$, $\tr(\Delta_{k}\pi)\approx 1$ if $\pi\in S_{k+1}^{\text{circ}}$ and $\tr(\Delta_{k}\pi)\approx 0$ otherwise. The lemma follows from \eqref{equ: O component 1} and triangle inequality. 
\end{proof}

\subsection{Proof of Lemma \ref{lem: round spectrum}}\label{sec: round spectrum}
\begingroup
\def\thetheorem{\ref{lem: round spectrum}}
\begin{lemma}
    Let $\ba=(a_1,\cdots, a_m)$ be a spectrum and $\psi_1,\cdots, \psi_m$ be $m$ Haar random states. With probability at least $0.99$, $\rho\coloneqq \sum_{r=1}^m a_r\psi_r$ is $\mathcal{O}(m\sqrt{\ln m}/\sqrt{d})$-close to a state with spectrum $\ba$ in trace distance.
\end{lemma}
\addtocounter{theorem}{-1}
\endgroup
\begin{proof}
    Let $\Psi$ be a $d\times m$ matrix whose $r$-th column is $\ket{\psi_r}$. Let $A_0$ be the $m\times m$ diagonal matrix with diagonal $\ba$. Then $\rho=\Psi A_0\Psi^\dagger$. 
    Define the Gram matrix $G\coloneqq \Psi^\dagger \Psi$.
    \begin{equation}
        \norm{G-\Id}_F^2 = \sum_{i\neq j} \abs{\braket{\psi_i}{\psi_j}}^2.
    \end{equation}
    By \Cref{lem: haar concentration inner product} and union bound, with probability at least $1-2m^2\exp(-d\delta/2)$, $\norm{G-\Id}_F^2\leq \frac{m^2}{d}+m^2\delta$. Setting $\delta=\frac{2\ln(200m^2)}{d}$, we obtain that with probability at least $0.99$, $\norm{G-\Id}_F^2=\mathcal{O}(\frac{m^2\ln m}{d})$ and $\norm{G-\Id}_F=\mathcal{O}(m\sqrt{\ln m}/\sqrt{d})$

    Write the $QR$ decomposition $\Psi=QR$, where $Q$ is a $d\times d$ unitary and $R$ is a $d\times m$ upper triangular matrix. Since $d>m$, $R$ consists of a $m\times m$ upper triangular matrix $R_0$ and a $(d-m)\times m$ zero matrix. Since $G=\Psi^\dagger \Psi=R^\dagger R=R_0^\dagger R_0$ and $R_0$ is upper triangular, $R_0^\dagger R_0$ is the Cholesky decomposition of $G$. According to perturbation bounds for Cholesky decomposition \cite[Theorem 1.1]{sunPerturbationBoundsCholesky1991}, $\norm{R_0-\Id_m}_F\leq \sqrt{2}\norm{G-\Id_m}_F$. 

    Let $A$ be the $d\times d$ diagonal matrix with diagonal $(a_1,\cdots, a_m, 0, 0, \cdots)$. Define $\sigma\coloneqq QAQ^{\dagger}$.
    Then $\sigma$ is a state with spectrum $\ba$. $\norm{\rho-\sigma}_1=\norm{Q(RA_0R^\dagger -A)Q^\dagger}_1=\norm{RA_0R^\dagger -A}_1=\norm{R_0A_0R_0^\dagger -A_0}_1$. Write $E=R_0-\Id$. By Holder's inequality $\norm{XY}_1\leq \norm{X}_F\norm{Y}_F$, we obtain
    \begin{align*}
        \norm{R_0A_0R_0^\dagger -A_0}_1&\leq \norm{EA_0}_1+\norm{A_0E}_1+\norm{EA_0E}_1\tag{triangle inequality}\\
        &\leq 2\norm{E}_F\norm{A_0}_F+\norm{EA_0}_F\norm{E}_F\tag{Holder's inequality}\\
        &\leq 2\norm{E}_F+\norm{E}_F^2 \tag{$\norm{A_0}_F\leq 1$, $\norm{EA_0}_F\leq \norm{A_0}_F$},
    \end{align*}
    where $\norm{EA_0}_F\leq \norm{A_0}_F$ because $A_0$ is a diagonal matrix with all entries in $[0, 1]$. Hence, with probability at least $0.99$, $\norm{\rho-\sigma}_1\leq \mathcal{O}(m\sqrt{\ln m}/\sqrt{d})$. 
\end{proof}

\section*{Acknowledgements}
We thank Weiyuan Gong, Xinyu Tan, Zhihan Zhang, Jonas Haferkamp, and Sitan Chen for helpful discussions. This work is supported by the National Natural Science Foundation of China (Grants No. T2225008, No. T24B2002, and No.~12475023), the Innovation Program for Quantum Science and Technology (No. 2021ZD0302203), Tsinghua University Dushi Program, the Shanghai Qi Zhi Institute Innovation Program SQZ202318, and a startup funding from YMSC.

\bibliographystyle{plainurl}

\end{document}